\documentclass[11pt,finalcls,onecolumn]{IEEEtran}

\RequirePackage{graphicx}
\RequirePackage{epstopdf}
\RequirePackage{epsfig}
\RequirePackage{amsmath}
\RequirePackage{amssymb}
\RequirePackage{float}
\RequirePackage{url}
\RequirePackage{tikz}
\RequirePackage{pgfplots}
\RequirePackage{grffile}
\RequirePackage{soul}
\RequirePackage{comment}

\usetikzlibrary{positioning, calc,arrows,plotmarks}
\pgfplotsset{compat=newest}
\usepgfplotslibrary{patchplots}
\pgfplotsset{plot coordinates/math parser=false}

\newtheorem{definition}{Definition}

\newtheorem{theorem}{Theorem}

\newtheorem{lemma}[theorem]{Lemma}
\newtheorem{remark}{Remark}
\newtheorem{example}{Example}

\newtheorem{construction}{Construction}

\begin{document}
\title{LDPC Codes with Local and Global Decoding}
\author{
\vspace*{0.8cm} Eshed Ram \qquad Yuval Cassuto\\
Andrew and Erna Viterbi Department of Electrical Engineering \\
Technion -- Israel Institute of Technology, Haifa 32000, Israel\\
E-mails: \{s6eshedr@campus, ycassuto@ee\}.technion.ac.il
}

\maketitle

\begin{abstract}
This paper presents a theoretical study of a new type of LDPC codes that is motivated by practical storage applications. LDPCL codes (suffix L represents locality) are LDPC codes that can be decoded either as usual over the full code block, or locally when a smaller sub-block is accessed (to reduce latency). LDPCL codes are designed to maximize the error-correction performance vs. rate in the usual (global) mode, while at the same time providing a certain performance in the local mode. We develop a theoretical framework for the design of LDPCL codes. Our results include a design tool to construct an LDPC code with two data-protection levels: local and global. We derive theoretical results supporting this tool and we show how to achieve capacity with it. A trade-off between the gap to capacity and the number of full-block accesses is studied, and a finite-length analysis of ML decoding is performed to exemplify a trade-off between the locality capability and the full-block error-correcting capability.\footnote{Part of the results of this paper were presented at the 2018 International Symposium on Information Theory.}
\end{abstract}


\section{Introduction}
\label{Sec:Intro}
Low-density parity-check (LDPC) codes and their low-complexity iterative decoding algorithm \cite{Gala62} are a powerful method to achieve reliable communication and storage with rates that approach Shannon's theoretical limit. Due to their efficient encoding and decoding algorithms, communication applications such as WiFi, DVB, and Ethernet had adopted this family of linear block codes. 
When used in data-storage applications, unlike in communications, retransmissions are not possible, and any decoding failure implies data loss; hence strong LDPC codes need to be provisioned for extreme data reliability. Another key feature of modern storage devices is fast access, i.e., low-latency and high-throughput read operations. 
However, high data reliability forces very large block sizes and high complexity, and thus degrades the device's latency and throughput. This inherent conflict motivates a coding scheme that enables fast read access to small (sub) blocks with modest data protection and low complexity, while in case of failure providing a high data-protection "safety net" in the form of decoding a stronger code over a larger block. Our objective in this paper is to design LDPC codes to operate in such a {\em multi-sub-block coding scheme}, where error-correction performance (vs. rate) is maximized in both the sub-block and full-block modes.

Formally, in a multi sub-block coding scheme, a code block of length $N$ is divided into $M$ sub-blocks of length $n$ (i.e., $N=Mn$). Each sub-block is a codeword of one code, and the concatenation of the $M$ sub-blocks forms a codeword of another (stronger) code. 
We define a new type of LDPC codes we call {\em LDPCL codes}, where the suffix 'L' points to the code's {\em local} access to its sub-blocks. 
The LDPCL code is designed in such a way that each of the sub-blocks (of length $n$) can be decoded independently of the other sub-blocks (local decoding), and in addition the full block of length $Mn$ can be decoded (global decoding) when local decoding fails. From application perspective, $M$ sets the ratio between an encoding unit and a local-decoding unit. For example, in a distributed-storage application, a data unit of $Mn$ bits is distributed across $M$ nodes, and each node requires access to its local sub-block of $n$ bits. In non-volatile storage devices, $M$ may be the ratio between the write block size and the read (sub-)block size.   

Our theoretical results on analysis and construction of LDPCL codes lie upon the definition of the code through two distinct degree-distribution pairs. The {\em local degree distribution} specifies the connections between sub-block variable nodes and their local check nodes, while the {\em joint degree distribution} governs the connection of the global check nodes to variable nodes in the full block. The key challenge is to design the local and joint distributions such that both the local and the global (=local+joint composition) codes perform well. 
In particular, the density-evolution analysis shows an inherent asymmetry between the local and joint codes. These asymmetries need to be addressed to make the analysis work.

The key promise of codes designed with this paper's tools is the low complexity of local sub-block decoding compared to the complexity of global full-block decoding. The complexity savings come from two sources: 1) a sub-block is factor $M$ smaller than the full block, and 2) the local code is designed with lower correction capability, allowing local degree distributions with lower node degrees. In applications where most decoding instances succeed locally (for example non-volatile memories with large error-rate variability), the average decoding complexity across instances will be close to the low complexity of the local code.     

\subsection{Related Work}
\label{Sub:relate}
Earlier work, such as~\cite{CassHemo17} and~\cite{Hass01,HanMont07,BlauHetz16}, addressed the design of Reed-Solomon and related algebraic codes in multi sub-block schemes. While that prior work attests to the importance of the multi sub-block scheme, designing LDPC codes for it requires new tools and methods. 
Another related family of codes are LDPC codes with incremental redundancy, which \cite{WangRangWesel17} shows how to use without feedback, albeit without a notion of local decoding. 
\cite{Li17} proposed global coupling for LDPC codes that results in codes with similar structure, but did not address the local-decoding performance, and focused on algebraic structured codes over non-binary alphabets. In this paper we consider binary codes, focus on asymptotic results, showing how to approach capacity while guaranteeing a certain degree of local-decoding performance.

In the context of code structure, the code ensembles presented in this paper are related to bi-layer LDPC codes proposed for other applications, such as communication over the relay channel \cite{RazaYu07, EzriGast06} and coding with side information \cite{WeinMart09}. In \cite{EzriGast06}, the two LDPC codes are designed independently, which for the present application cannot yield optimal codes because global decoding performance depends on both the local and the joint degree distributions. In \cite{WeinMart09}, one of the codes is an LDGM code, and only regular codes are used. 
The work with the most relevant results to this paper is \cite{RazaYu07}, but there are several important distinctions between their codes and ours. First, our codes consist of multiple sub-blocks that can be decoded independently, in contrast to the codes in \cite{RazaYu07} where there is only one code block. The differences in structure and in the target application also imply completely different decoders and design objectives for the codes.  
Another distinction is that the code construction in \cite{RazaYu07} requires specifying the full product degree distribution of the two layers, which only lends itself to unwieldy design through linear programming. In contrast, our constructions build on a new analysis of the cooperation between the local and joint codes, and as a result we are able to provide degree distributions that provably approach capacity for arbitrary local and global thresholds. Critical to the results of this paper is our choice to specify the code ensembles through two separate degree-distribution pairs (local and joint), which are designed jointly to approach capacity.

\subsection{Contributions}
\label{Sub:contributions}

The main scope in this paper is the binary erasure channel (BEC). The analysis and construction techniques can be extended to other channels (such as the binary-input additive Gaussian-noise channel), but we do not pursue those here. Our main contributions are as follows. 
\begin{enumerate}
	\item An LDPC-type solution for multi sub-block coding in Section~\ref{Sec:Ensembles}, and an asymptotic analysis of the suggested codes under belief-propagation (BP) decoding in Section~\ref{Sec:2D DE}. 
	\item An easy-to-use tool to construct capacity-approaching codes for the BEC in Section~\ref{Sec:C1} (code optimization using linear programming can be easily formulated, but we omit the details).  
	\item A study of the trade-off between the gap to capacity to the number of decoding iterations accessing the full block in Section~\ref{Sec:N_JI}. We also suggest an optimal scheduling scheme that minimizes the access to full-block bits and thus saves time and communication costs when the sub-blocks are distributed. 
	\item A finite-length analysis for ML decoding of (regular) multi sub-block LDPC codes (Section~\ref{Sec:Finite}) that exemplifies the trade-off between sub-block access and full-block performance. 
\end{enumerate}

\section{Preliminaries}
\label{Sec:Pre}

\subsection{LDPC Codes: Review from \cite{RichUrb08}}
\label{sub:LDPC}
A linear block code is an LDPC code if it has at least one parity-check matrix that is sparse, i.e., the number of 1's in $H$ is linear in the block length. This sparsity enables low-complexity decoding algorithms. 
Every parity-check matrix $H$ can be represented by a bipartite graph, called a Tanner graph, with  nodes partitioned to variable nodes and check nodes; there exists an edge between check node $i$ and variable node $j$, if and only if $H_{ij}=1$ (this paper focuses on binary linear codes, but this representation can be generalized). The fraction of variable (resp. check) nodes in a Tanner graph with degree $i$ is denoted by $\Lambda_i$ (resp. $\Omega_i$), and the fraction of edges connected to variable (resp. check) nodes of degree $i$ is denoted by $\lambda_i$ (resp. $\rho_i$); $\Lambda_i$ and $\Omega_i$ are called node-perspective degree distributions, and $\lambda_i$ and $\rho_i$ are called edge-perspective degree distributions. 

The degree-distribution polynomials associated to a Tanner graph are given by
\begin{subequations}
\begin{align} 
\label{eq:dd poly1}
&\Lambda(x)=\sum_{i}\Lambda_{i} x^i,    \qquad  \lambda(x)=\sum_{i}\lambda_{i} x^{i-1},\quad x \in [0,1],\\
\label{eq:dd poly2}
& \Omega(x)=\sum_{i}\Omega_{i}x^i ,\qquad \rho(x) = \sum_{i}\rho_{i} x^{i-1},\quad x \in [0,1].
\end{align}
The node-perspective and edge-perspective polynomials are related through 
\begin{align} \label{eq:node-edge}
&\Lambda(x)= \tfrac{\int_0^x \lambda(t)\mathrm{d}t}{\int_0^1 \lambda(t)\mathrm{d}t}, \hspace*{0.8cm}
\Omega(x) = \tfrac{\int_0^x \rho(t)\mathrm{d}t}{\int_0^1 \rho(t)\mathrm{d}t},\quad x\in [0,1], \\
\label{eq:edge-node}
&\lambda(x)= \tfrac{\Lambda'(x)}{\Lambda'(1)}, \hspace*{1.4cm}
\rho(x)   = \tfrac{\Omega'(x)}{\Omega'(1)},\quad \hspace*{0.45cm} x\in [0,1],
\end{align}
\end{subequations}
where the operator $'$ stands for the function's derivative. 

\subsection{The Sub-Blocked Tanner Graph}
\label{sub:Tanner}

We define an LDPCL code of length $N=Mn$ through a sub-blocked Tanner graph. 
In this sparse graph, the variable nodes are divided to $M$ disjoint sets (sub-blocks) of size $n$ each, and the check nodes are divided into two disjoint sets: {\em local check nodes} and {\em joint check nodes}. 
The graph construction is constrained such that each local check node is connected only to variable nodes that are in the same sub-block of length $n$; the joint check-node connections have no constraints. 
The edges of the graph are partitioned into two sets as well: edges connecting variable nodes to local check nodes are called {\em local edges}, and edges connecting variable nodes to joint check nodes are called {\em joint edges}. Finally, the local (resp. joint) degree of a variable node is the number of local (resp. joint) edges emanating from it. 

\begin{example} \label{Ex:MB Tanner}
A sub-blocked Tanner graph with $ M=3 $ sub-blocks -- each of length $n=6$ -- is illustrated in Figure~\ref{Fig:MB Tanner}. Local (resp. joint) checks contain an 'L' (resp. 'J') label.
\end{example}

\begin{figure}[!h]

	\begin{center}
     	\begin{tikzpicture}[scale=0.3,>=latex]\label{Tikz:2-side tanner}
     	\tikzstyle{cnode}=[rectangle,draw,fill=gray!70!white,minimum size=2mm]
     	\tikzstyle{vnode}=[circle,draw,fill=gray!70!white,minimum size=2mm]
		\pgfmathsetmacro{\x}{10}
		\pgfmathsetmacro{\w}{1.5}
		\pgfmathsetmacro{\y}{4}
     	\foreach \m in {1,2,3}
     	{	
     		\foreach \v in {1,...,6}
     		{	
     			\node[vnode] (v\m\v) at (\m*\x-\x+\v*\w-\w,0) {};	
     		}
     		\foreach \c in {1,...,3}
     		{	
     			\node[cnode] (c\m\c) at (\m*\x-\x+\c*2*\w-1.5*\w,-\y) {\footnotesize L};	
     		}
     		\draw[thick] (v\m1.south)--(c\m1.north);
     		\draw[thick] (v\m1.south)--(c\m2.north);
     		\draw[thick] (v\m2.south)--(c\m1.north);	
     		\draw[thick] (v\m2.south)--(c\m2.north);
     		\draw[thick] (v\m3.south)--(c\m1.north);	
     		\draw[thick] (v\m3.south)--(c\m2.north);
     		\draw[thick] (v\m4.south)--(c\m1.north);	
     		\draw[thick] (v\m4.south)--(c\m2.north);
     		\draw[thick] (v\m4.south)--(c\m3.north);	
     		\draw[thick] (v\m5.south)--(c\m2.north);
     		\draw[thick] (v\m5.south)--(c\m3.north);	
     		\draw[thick] (v\m6.south)--(c\m2.north);
     		\draw[thick] (v\m6.south)--(c\m3.north);	
     	}		
     	\node[cnode] (cJ1) at (1*\x-\x+4*\w-1.5*\w,\y) {\footnotesize J};
     	\node[cnode] (cJ2) at (2*\x-\x+4*\w-1.5*\w,\y) {\footnotesize J};	
     	\node[cnode] (cJ3) at (3*\x-\x+4*\w-1.5*\w,\y) {\footnotesize J};	
     	\draw[thick] (cJ1.south)--(v12.north);	
     	\draw[thick] (cJ1.south)--(v13.north);	
     	\draw[thick] (cJ1.south)--(v16.north);	
     	\draw[thick] (cJ1.south)--(v21.north);	
     	\draw[thick] (cJ1.south)--(v22.north);	
     	\draw[thick] (cJ1.south)--(v24.north);
     	
     	\draw[thick] (cJ2.south)--(v13.north);	
     	\draw[thick] (cJ2.south)--(v15.north);	
     	\draw[thick] (cJ2.south)--(v23.north);	
     	\draw[thick] (cJ2.south)--(v25.north);	
     	\draw[thick] (cJ2.south)--(v26.north);
     	
     	\draw[thick] (cJ3.south)--(v15.north);	
     	\draw[thick] (cJ3.south)--(v22.north);	
     	\draw[thick] (cJ3.south)--(v26.north);	
     	\draw[thick] (cJ3.south)--(v32.north);	
     	\draw[thick] (cJ3.south)--(v33.north);	
     	\draw[thick] (cJ3.south)--(v34.north);
     	\draw[thick] (cJ3.south)--(v36.north);
     	\end{tikzpicture}
     \end{center}

     \caption{\label{Fig:MB Tanner}
      Example of a sub-blocked Tanner graph  with $ M=3 $ and $n=6$. Local (resp. joint) checks contain an 'L' (resp. 'J') label.}

\end{figure}
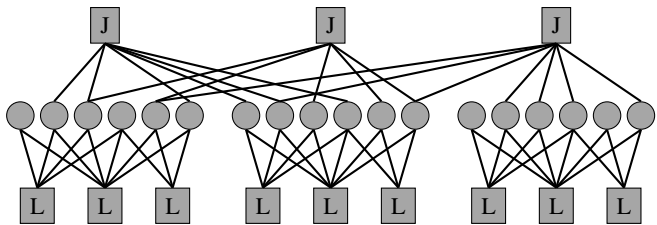

We denote by $\Lambda_{L,i}$ the fraction of variable nodes with local degree $i$, and by $\Omega_{L,i}$ the fraction of local check nodes with degree $i$. Similarly, $\lambda_{L,i}$ designates the fraction of local edges connected to a variable node with local degree $i$, and $\rho_{L,i}$ designates the fraction of local edges connected to a local check node of degree $i$. We call $\left(\Lambda_{L,i},\Omega_{L,i},\lambda_{L,i},\rho_{L,i} \right)$ local degree distributions. Note that we do not distinguish between local degree distributions of different sub-blocks, and we assume that they are the same in all sub-blocks (but the instances drawn from the distributions are in general different between the sub-blocks).
The joint degree distributions $\left(\Lambda_{J,i},\Omega_{J,i},\lambda_{J,i},\rho_{J,i}\right)$ are defined similarly but with an important difference: we allow some variable nodes to have joint degree $0$ or $1$. 
Joint degree $ 0 $ increases the rate without compromising decoding performance (as would happen in ordinary LDPC codes with a single degree distribution). Due to its importance, we will use in the rest of the paper $P_0$ to denote the coefficient $\Lambda_{J,0}$.

The local and joint degree-distribution polynomials $\Lambda_L(\cdot),\lambda_L(\cdot),\Omega_L(\cdot),\rho_L(\cdot)$ and $\Lambda_J(\cdot),\lambda_J(\cdot),\Omega_J(\cdot),\rho_J(\cdot)$ are defined similarly to the degree-distribution polynomials for ordinary $LDPC$ codes in \eqref{eq:dd poly1}-\eqref{eq:dd poly2}. Known relations between node-perspective and edge-perspective polynomials hold for the local polynomials. However, since some variable nodes may have a joint degree of zero, the equation that describes $\Lambda_J$ in terms of $\lambda_J$ changes to
\begin{align} \label{eq:joint edge2node}
\Lambda_J(x)= P_0+(1-P_0)\tfrac{\int_0^x \lambda_J(t)\mathrm{d}t}{\int_0^1 \lambda_J(t)\mathrm{d}t}, \quad x \in [0,1].
\end{align}

\subsection{LDPCL Ensembles}
\label{Sec:Ensembles}
In this sub-section we define the ensembles of sub-blocked Tanner graphs: the LDPCL ensembles. 
These ensembles have six parameters: $M,n,\Lambda_L(\cdot),\Lambda_J(\cdot),\Omega_L(\cdot),$ and $\Omega_J(\cdot)$. $M$ is the locality parameter that sets the number of sub-blocks in a code block, $n$ is the sub-block length, and $\Lambda_L(\cdot),\Lambda_J(\cdot),\Omega_L(\cdot),\Omega_J(\cdot)$ are the node-perspective degree-distribution polynomials; this ensemble is denoted by $LDPCL(M,n,\Lambda_L,\Omega_L,\Lambda_J,\Omega_J)$. If we refer to an LDPCL ensemble through its edge-perspective degree-distribution polynomials, then we write  $LDPCL(M,n,\lambda_L,\rho_L,\lambda_J,\rho_J,P_0)$ (when using the edge-perspective notation, one must specify $P_0$ as well).

The sampling process from the $LDPCL(M,n,\Lambda_L,\Lambda_J,\Omega_L,\Omega_J)$ ensemble is a as follows. First, $M$ Tanner graphs are sampled independently from the $LDPC(n,\Lambda_L,\Omega_L)$ ensemble. These local graphs are positioned next to each other without inter-connections to create a Tanner graph with $ Mn $ variable nodes. Another Tanner graph is then sampled from the $LDPC(Mn,\Lambda_J,\Omega_J)$ ensemble. The latter joint graph is flipped, its $Mn$ variable nodes are randomly permuted, and merged with the $Mn$ variable nodes of the $M$ local graphs to create a sub-blocked Tanner graph. 
We add the random permutation to force statistical independence between the local and joint degrees of the variable nodes. The design rate of an $LDPCL(M,n,\Lambda_L,\Lambda_J,\Omega_L,\Omega_J)$ ensemble is given by
\begin{align}
R &=  1 - \frac{\Lambda_L'(1)}{\Omega_L'(1)} -  \frac{\Lambda_J'(1)}{\Omega_J'(1)} \notag \\
\label{eq:design rate2}
&=  1 - \frac{\int_0^1\rho_L(x)\mathrm{d}x}{\int_0^1\lambda_L(x)\mathrm{d}x} -  \frac{\int_0^1\rho_J(x)\mathrm{d}x}{\int_0^1\lambda_J(x)\mathrm{d}x}\left( 1-P_0\right).
\end{align}
We can see in \eqref{eq:design rate2}, that setting $P_0>0$ allows increasing the code rate, which we later find crucial in our constructions.

\section{Decoding Analysis}
\label{Sec:2D DE}
In this section we suggest a decoding strategy for LDPCL codes, and analyze its performance. Our 
ultimate goal (in Section~\ref{Sec:C1}) is to provide a design tool for building sub-blocked LDPC codes: given two noise levels -- local $ \epsilon_L $ and global $ \epsilon_G $ -- produce LDPCL degree distributions such that sub-block access provides the correction capability to tolerate $ \epsilon_L $, and full-block access provides the global correction capability to tolerate $ \epsilon_G $. The derivations in this section lay the theoretical infrastructure needed to show the optimality of our constructions (i.e., capacity achieving in Section~\ref{Sec:C1}). 

To take advantage of the locality structure of the Tanner graphs described above, the suggested decoding algorithm will operate in two modes: local mode and global mode. In the local mode, the decoder tries to decode a sub-block of length $n$ using belief propagation (BP) on the local Tanner graph.
If it succeeds (e.g. due to high SNR), then the information is passed to the user for fast access.
If the decoder meets a failure criterion (e.g., getting stuck or reaching maximum number of iterations), then it enters the global mode where it tries to decode the entire code block (of length $N=Mn$) using BP on the complete multi sub-block Tanner graph.
In this section we assume that the message scheduling in the global mode is a flooding schedule: in the first step of a global decoding iteration, the variable nodes send messages to the local and joint check nodes in parallel, and in the second step the local and global check nodes send their messages back to the variable nodes 
(later in Section~\ref{Sec:N_JI}, we change the schedule from flooding to be more locality aware). 

Consider a BEC channel with erasure probability $ \epsilon \in(0,1)$.
In the local mode, the asymptotic (as $n \to \infty$) analysis of the decoding algorithm is identical to the asymptotic analysis of ordinary LDPC codes. Specifically, in the limit where $n\to \infty$, there exists a local decoding threshold $\epsilon^*_L$, such that if $\epsilon<\epsilon^*_L$, the decoder will resolve the desired sub-block in the local mode with probability converging to $1$. If $\epsilon>\epsilon^*_L$, the decoder will fail in the local mode with probability converging to $1$. 
$\epsilon^{*}_L$ can be calculated numerically via
\begin{align} \label{eq:1D th numeric}
\epsilon^{*}_L = \inf_{(0,1]} \frac{x}{\lambda_L(1-\rho_L(1-x))}.
\end{align}

In the global mode we have the following.
\begin{theorem}\label{th:2D DE}
Consider a random element from the \\$LDPCL(M,n,\Lambda_L,\Lambda_J,\Omega_L,\Omega_J)$ ensemble. Let $x_l(\epsilon)$ and $y_l(\epsilon)$ denote the probability that a local and joint edge, respectively, carries a variable-to-check erasure message after $l$ BP iterations over the $BEC(\epsilon)$ as $n \to \infty$. Then,
\begin{subequations}
\begin{align}
\label{eq:2D DE1}
&x_l(\epsilon) = \epsilon \cdot \lambda_L\left(1- \rho_L\left(1-x_{l-1}(\epsilon) \right)\right)\cdot \Lambda_J\left(1- \rho_J\left(1-y_{l-1}(\epsilon) \right)\right), \quad l \geq 0,\\
\label{eq:2D DE2}
&y_l(\epsilon) = \epsilon \cdot \Lambda_L\left(1- \rho_L\left(1-x_{l-1}(\epsilon) \right)\right) \cdot \lambda_J\left(1- \rho_J\left(1-y_{l-1} (\epsilon) \right)\right),\quad l \geq 0,\\
\label{eq:2D DE3}
&x_{-1}(\epsilon)=y_{-1}(\epsilon)=1.
\end{align}
\end{subequations}
\end{theorem}

\begin{proof}
See Appendix~\ref{App:2D DE} for the full proof. Figure~\ref{Fig:DE Eq} graphically illustrates equations \eqref{eq:2D DE1}--\eqref{eq:2D DE2}: in the center diagram the right outgoing edge carries the message in \eqref{eq:2D DE1} to a local check and the left outgoing edge carries the message in \eqref{eq:2D DE2} to a joint check. 
\end{proof}

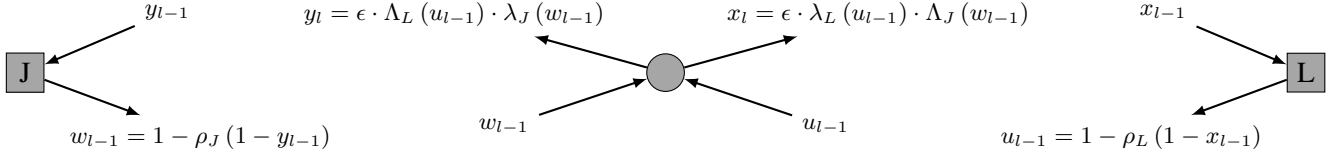
\begin{figure}[!h]
	
	\begin{center}
		\begin{tikzpicture}[scale=0.3,>=latex]\label{Tikz:DE eq}
		\tikzstyle{cnode}=[rectangle,draw,fill=gray!70!white,minimum size=5mm]
		\tikzstyle{vnode}=[circle,draw,fill=gray!70!white,minimum size=5mm]
		\pgfmathsetmacro{\x}{20}
		\pgfmathsetmacro{\w}{4*\x}
		\pgfmathsetmacro{\y}{4}
		
		\node[cnode] (cJ) {J};
		\node [above right = 3mm and 0.6*\x mm of cJ] (incJ) {\footnotesize $ y_{l-1} $};
		\node [below right = 3mm and 0.1*\x mm of cJ] (outcJ) {\footnotesize $ w_{l-1}=1-\rho_J\left (1-y_{l-1}\right ) $};
		\draw[->,thick] (incJ)--(cJ) ;
		\draw[->,thick] (cJ)--(outcJ) ;
		
		\node[vnode,right=\w mm of cJ] (v) {};
		\node [above right = 3mm and 0.25*\x mm of v]  (outv2) {\footnotesize $ x_{l}=\epsilon\cdot\lambda_L\left (u_{l-1} \right )\cdot\Lambda_J\left (w_{l-1} \right)$};
		\node [below left = 3mm and 0.75*\x mm of v] (invJ) {\footnotesize $ w_{l-1} $};
		\node [above left = 3mm and 0.25*\x mm of v] (outv1) {\footnotesize $ y_{l}=\epsilon\cdot\Lambda_L\left (u_{l-1} \right )\cdot\lambda_J\left (w_{l-1} \right)$};
		\node [below right = 3mm and 0.75*\x mm of v] (invL) {\footnotesize $ u_{l-1} $};;
		\draw[->,thick] (invL)--(v) ;
		\draw[->,thick] (invJ)--(v) ;
		\draw[->,thick] (v)--(outv1) ;
		\draw[->,thick] (v)--(outv2) ;
		
		\node[cnode,right=\w mm of v] (cL) {L};
		\node [above left = 3mm and 0.6*\x mm of cL] (incL) {\footnotesize $ x_{l-1} $};
		\node [below left = 3mm and 0.1*\x mm of cL] (outcL) {\footnotesize $ u_{l-1}=1-\rho_L\left (1-x_{l-1}\right ) $};
		\draw[->,thick] (incL)--(cL);
		\draw[->,thick] (cL)--(outcL) ;
		
		\end{tikzpicture}
	\end{center}
	
	\caption{\label{Fig:DE Eq} Illustration of the DE equations \eqref{eq:2D DE1}--\eqref{eq:2D DE2}. See Appendix~\ref{App:2D DE} for a detailed proof.}
	
\end{figure}

To simplify notations, $\epsilon$ will be omitted from now on from $x_l(\epsilon)$ and $y_l(\epsilon)$ if it is clear from the context. 

\begin{remark} \label{remark:asymmetry}
Although $x_l$ and $y_l$ in \eqref{eq:2D DE1}-\eqref{eq:2D DE2} seem symmetric to each other, it is not necessarily true since we allow variable nodes to have joint degrees $0$ ($P_0>0$) or $1$ ($\lambda_J(0)>0$), while their local degrees are forced to be strictly greater then $1$. This asymmetry has a crucial effect on the global decoding process which is explained and detailed in Section~\ref{sub:Threshold}.
Symmetry does hold in the special and less interesting case where $\rho_L(\cdot)=\rho_J(\cdot)$ and $\lambda_L(x)=x^{l_L-1},\lambda_J(x)=x^{l_J-1}$, in which case, for every iteration $ l\geq 0 $, $y_l =x_l=\epsilon \lambda(1-\rho(1-x_{l-1})),$ where, $\rho(x)\triangleq\rho_L(x)$, $\lambda(x)\triangleq x^{l_L+l_J-1}$. Thus if we use identical degree-distributions for local and joint check nodes and we force all variable nodes to have local and joint regular degrees, then the 2D-DE equations in \eqref{eq:2D DE1}-\eqref{eq:2D DE2} degenerate to the already known 1D-DE equation. However, codes falling under this special case are less interesting because they are sub-optimal in their rates and restricted in their thresholds. 

\end{remark}

\subsection{Threshold}
\label{sub:Threshold}
We now study the asymptotic global threshold of LDPCL ensembles. We prove that there exists a global decoding threshold denoted by $\epsilon^*_G$, characterize it, and provide a method to numerically calculate it. 
The results in this sub-section are the basis for code design we address in Section~\ref{Sec:C1}.

Define 
\begin{subequations}
\begin{align} 
\label{eq:f}
f(\epsilon,x,y)=\epsilon\,\lambda_L\left(1- \rho_L\left(1-x \right)\right) \Lambda_J\left(1- \rho_J\left(1-y\right)\right),\quad x,y,\epsilon \in [0,1] \\
\label{eq:g}
g(\epsilon,x,y)=\epsilon\,\Lambda_L\left(1- \rho_L\left(1-x \right)\right)  \lambda_J\left(1- \rho_J\left(1-y\right)\right) ,\quad x,y,\epsilon \in [0,1]
\end{align}
\end{subequations}

such that \eqref{eq:2D DE1}-\eqref{eq:2D DE3} can be re-written as
\begin{align} \label{eq:2D DE with f,g}
\begin{split}
&x_l = f\left(\epsilon,x_{l-1},y_{l-1}\right) ,\quad l \geq 0\\
&y_l = g\left(\epsilon,x_{l-1},y_{l-1}\right) ,\quad l \geq 0\\
&x_{-1} =y_{-1} =1.
\end{split}
\end{align}

\begin{lemma} \label{lemma:monotonicity}
The functions $f$ and $g$ are monotonically non-decreasing in all of their variables.
\end{lemma}
\begin{proof}
Since the images of $\lambda_L(\cdot),\Lambda_L(\cdot),\lambda_J(\cdot),\Lambda_J(\cdot),\rho_L(\cdot)$ and $\rho_J(\cdot)$ lie in $[0,1]$, then $f$ and $g$ are monotonically non-decreasing in $\epsilon \in [0,1]$. The proof for $ x,y $ is similar and is left as an exercise.
\end{proof}

\begin{definition}
Let $\epsilon \in (0,1)$. We say that $(x,y) \in [0,1]^2$ is an $(f,g)$-fixed point if
\begin{align} \label{eq:f,g fixed}
\begin{pmatrix}
x\\y
\end{pmatrix}
=
\begin{pmatrix}
f(\epsilon,x,y)\\
g(\epsilon,x,y)
\end{pmatrix}.
\end{align}
\end{definition}

Clearly, for every $\epsilon \in (0,1)$, $(x,y)=(0,0)$ is a trivial $(f,g)$-fixed point. However, it is not clear yet if there exists a non-trivial $(f,g)$-fixed point. In particular, we ask: for which choices of $\epsilon,\; \lambda_L,\rho_L,\lambda_J,\rho_J$ and $P_0$ there exists a non-trivial $(f,g)$-fixed point? The following lemmas will help answering this question.

\begin{lemma} \label{lemma:fix pt}
Let $\epsilon \in (0,1)$, and let $(x,y) \in [0,1]^2$ be an $(f,g)$-fixed point. Then,
\begin{enumerate}
\item \label{item: fix pt lemma0} $x=0$ implies $y =0$, and if $P_0=0$ or $\lambda_J(0)>0$, then $y=0$ implies $x =0$.
\item \label{item: fix pt lemma1} $(x,y) \in [0,\epsilon)^2$.
\item \label{item: fix pt lemma2} If $\{x_{l}\}_{l=0}^\infty$ and $\{y_{l}\}_{l=0}^\infty$ are defined by \eqref{eq:2D DE with f,g}, then
\begin{align} \label{eq:no jumps}
x_{l} \geq x , \quad y_{l} \geq y,\quad \forall l\geq 0.
\end{align}
\end{enumerate}
\end{lemma}

\begin{proof}
See Appendix~\ref{App:fix pt}
\end{proof}

\begin{remark} \label{remark: noting asymmetry}
Item 1 in Lemma~\ref{lemma:fix pt} expresses the asymmetry (discussed in Remark~\ref{remark:asymmetry}) between the local and joint sides during the decoding algorithm. 
\end{remark}

\begin{lemma} \label{lemma:mono of DE}
Let $x_{l} $ and $y_{l}$ be defined by \eqref{eq:2D DE with f,g} and let $0<\epsilon \leq \epsilon' < 1$. Then,
\begin{subequations}
\begin{align} \label{eq:mono of DE1}
\begin{split}
x_{l+1}(\epsilon) \leq x_{l}(\epsilon) ,\quad y_{l+1}(\epsilon) \leq y_{l}(\epsilon) , \quad \forall l \geq 0,
\end{split}
\end{align}
and
\begin{align} \label{eq:mono of DE2}
\begin{split}
x_{l}(\epsilon) \leq x_{l}(\epsilon') , \quad y_{l}(\epsilon) \leq y_{l}(\epsilon') , \quad \forall l \geq 0.
\end{split}
\end{align}
\end{subequations}
\end{lemma}
\begin{proof}
By mathematical induction on $l$ and by Lemma~\ref{lemma:monotonicity}. The details are left as an exercise.
\end{proof}

In view of \eqref{eq:2D DE with f,g}, it can be verified that for every iteration $ l\geq 0 $, $x_l(0) = y_l(0) = 0 ,\; x_l(1) = y_l(1) = 1$. Since $x_l$ and $y_l$ are bounded from below by $0$, then  Lemma~\ref{lemma:mono of DE} implies that the limits $\lim\limits_{l \to \infty}x_l(\epsilon)$ and $\lim\limits_{l \to \infty}y_l(\epsilon)$ exist. Thus we can define a global decoding threshold by
\begin{align} \label{eq:th op def A}
\epsilon^*_G = \sup \left\{\epsilon \in [0,1] \colon \lim\limits_{l \to \infty}y_l(\epsilon)=\lim\limits_{l \to \infty}x_l(\epsilon)=0 \right\}.
\end{align}
Note that from the continuity of $g$ in \eqref{eq:g}, item~\ref{item: fix pt lemma0} in Lemma~\ref{lemma:fix pt} implies that if $\lim\limits_{l \to \infty}x_l(\epsilon)=0 $, then  $\lim\limits_{l \to \infty}y_l(\epsilon)=0 $. Thus, \eqref{eq:th op def A} can be re-written as
\begin{align} \label{eq:th op def}
\epsilon^*_G = \sup \left\{\epsilon \in [0,1] \colon \lim\limits_{l \to \infty}x_l(\epsilon)=0 \right\}.
\end{align}

\begin{theorem}
\label{theorem:fix pt char}
Let
\begin{align}\label{eq:fix pt char}
\hat{\epsilon} = \sup \left\{ \epsilon \in [0,1]\colon \eqref{eq:f,g fixed}\text{ has no solution with }(x,y) \in (0,1]\times[0,1]\right\}.
\end{align}
Then, $\epsilon^*_G=\hat{\epsilon}$.
\end{theorem}

\begin{proof}
Let $\epsilon<\hat{\epsilon}$, and let $x(\epsilon)=\lim\limits_{l \to \infty}x_l(\epsilon),\;y(\epsilon)=\lim\limits_{l \to \infty}y_l(\epsilon)$. Taking the limit $l \to \infty$ in \eqref{eq:2D DE with f,g} yields that $(x(\epsilon),y(\epsilon))$ is an $(f,g)$-fixed point. In view of  \eqref{eq:fix pt char}, since $\epsilon<\hat{\epsilon}$, it follows that $x(\epsilon)=0$. From \eqref{eq:th op def} we have  $\epsilon < \epsilon^*_G$, for every $\epsilon<\hat{\epsilon}$; this implies that $\hat{\epsilon} \leq \epsilon^*_G$. 

For the other direction, let $\epsilon>\hat{\epsilon}$ and let $(z_1,z_2)$ be an $(f,g)$-fixed point such that $z_1>0$. Lemma~\ref{lemma:fix pt}-item~\ref{item: fix pt lemma2} implies that
\begin{align*}
x_l(\epsilon) \geq z_1 >0,\quad \forall l \geq 0,\\
\end{align*} 
thus $\lim\limits_{l \to \infty}x_l(\epsilon) >0$, where the existence of this limit is assured due to Lemma~\ref{lemma:mono of DE}; hence, $\epsilon > \epsilon^*_G$. Since this is true for all $\epsilon>\hat{\epsilon}$, then we deduce that $\hat{\epsilon} \geq \epsilon^*_G$ and complete the proof.
\end{proof}

We proceed by providing a numerical way to calculate the threshold of a given choice of $\Lambda_L,\Lambda_J,\Omega_L$ and $\Omega_J$. Define
\begin{align} \label{eq:q_L q_J}
q_L(x)\triangleq x\cdot \frac{\Lambda_L\left( 1- \rho_L\left(1-x \right)\right)}{\lambda_L\left( 1- \rho_L\left(1-x \right)\right)},\quad q_J(x)\triangleq x\cdot \frac{\Lambda_J\left( 1- \rho_J\left(1-x \right)\right)}{\lambda_J\left( 1- \rho_J\left(1-x \right)\right)},\quad x \in (0,1].
\end{align}

\begin{lemma} \label{lemma:q_L(0)}
$\lim_{x \to 0}q_L(x)=0$.
\end{lemma}

\begin{proof}
See Appendix~\ref{App:q_L(0)}.
\end{proof}

Since $q_L(1)=1$, Lemma~\ref{lemma:q_L(0)} and the intermediate-value theorem imply that for every $w \in (0,1]$, there exists $x\in (0,1]$ such that $q_L(x)=w$. Note that it is not true in general that $\lim_{x \to 0}q_J(x)=0$ (another evidence of the local-joint asymmetry); this limit may be infinite (for example the case $P_0>0$, $\rho_J(x)=x^3$ and $\lambda_J(x)=x^2$). 

\begin{definition} \label{def:q}
For every $y>0$ such that $q_J(y) \leq 1$ define
\begin{align} \label{eq:q def}
q(y)\triangleq\max \{ x \colon q_L(x)=q_J(y)\}.
\end{align}
\end{definition}

\begin{theorem}\label{th:numerical th}
Let $\lambda_L,\rho_L,\lambda_J,\rho_J$ be degree-distribution polynomials, let $P_0 \in [0,1]$, and let $\epsilon^*_G=\epsilon^*_G(\Lambda_L,\Lambda_J,\Omega_L,\Omega_J)$ be the global  threshold of the $LDPCL(M,n,\Lambda_L,\Lambda_J,\Omega_L,\Omega_J)$ ensemble. \\
If $P_0=0$ or $\lambda_J(0) >0$ , then
\begin{align} \label{eq:numerical th1}
\epsilon^*_G =  \inf\limits_{\substack{y \in (0,1] \\ q_J(y)\leq 1}} \frac{y}{g(1,q(y),y)}.
\end{align}
Else,
\begin{align} \label{eq:numerical th2}
\epsilon^*_G = \min\left\{ \inf\limits_{\substack{y \in (0,1] \\ q_J(y)\leq 1}} \frac{y}{g(1,q(y),y)},\;\;\frac1{P_0}\cdot \inf_{(0,1]}\frac{x}{\lambda_L\left(1- \rho_L\left(1-x \right)\right)}\right\}.
\end{align}
\end{theorem}

\begin{proof}
See Appendix~\ref{App:numerical th}.
\end{proof}

\begin{example} \label{ex:threshold}
Consider an LDPCL ensemble characterized by 
\begin{align*}
\lambda_L(x)=x, \quad \rho_L(x)=x^9 ,\quad \lambda_J(x)=0.3396x+0.6604x^4,\quad P_0=0.2667, \quad \rho_J(x)=x^9 .
\end{align*}
Using \eqref{eq:design rate2} and \eqref{eq:1D th numeric}, the design rate is $R=0.5571$ and the local decoding threshold is $\epsilon_L^*=0.1112$. By calculating the arguments in \eqref{eq:numerical th2}, the global decoding threshold is $\epsilon_G^*=\min\{0.35,0.4168\}=0.35$ (better rates are achieved in the next section, and these degree distributions are given to graphically exemplify the results derived so far). 
Figure~\ref{Fig:DE} illustrates the 2D-DE equations in \eqref{eq:2D DE1}-\eqref{eq:2D DE3} for three different erasure probabilities: $0.33,0.35,0.37$, from left to right, respectively. When the channel's erasure probability is $\epsilon=0.33$, there are no $(f,g)$-fixed points -- the decoding process ends successfully, and when $\epsilon=0.37$, there are two $(f,g)$-fixed points, $(0.335,0.3202)$ and $(0.2266,0.1795)$ -- the decoding process gets stuck at $(0.335,0.3202)$. When $\epsilon=0.35=\epsilon_G^*$, there is exactly one $(f,g)$-fixed point at $(0.27,0.237)$, and the dashed and dotted lines osculate.
\end{example}

\begin{remark}\label{remark:complexity}
As mentioned in Section~\ref{Sec:Intro}, the complexity advantage of local-decoding LDPCL codes over global-decoding ordinary LDPC codes comes from 1) a sub-block is factor $ M $ smaller than the full block, and 2) lower local node degrees. For the BEC, counting edges in the Tanner graph is a good approximation of the decoding complexity, thus we now perform a comparison between the number of edges in the local code, denoted by $ |E_L| $, and the number of edges in a full-block LDPC code, denoted by $ |E| $. It is known that if the number of variable nodes in the graph is $ n $ and their degree distribution is given by $ \lambda(\cdot) $, then the number of edges is given by $ n/\!\int_0^1 \lambda  $. Hence the ratio between the number of local and global edges is
\begin{align}\label{eq:complexity comparison}
\frac{|E_L|}{|E|}=\frac{n\cdot\left (\int_0^1\lambda_L(x)\mathrm{d}x\right )^{-1}}{nM\cdot\left (\int_0^1\lambda_J(x)\mathrm{d}x\right )^{-1}}=\frac1M\cdot\frac{\int_0^1\lambda_J(x)\mathrm{d}x}{\int_0^1\lambda_L(x)\mathrm{d}x}.
\end{align} 
For example, the code ensemble from Example~\ref{ex:threshold} above with $ M = 4$ sub-blocks yields $ \tfrac{|E_L|}{|E|} = 0.1509$ -- a reduction of almost 85\%.
\end{remark}
\begin{figure}[h!]
\begin{center}
\input{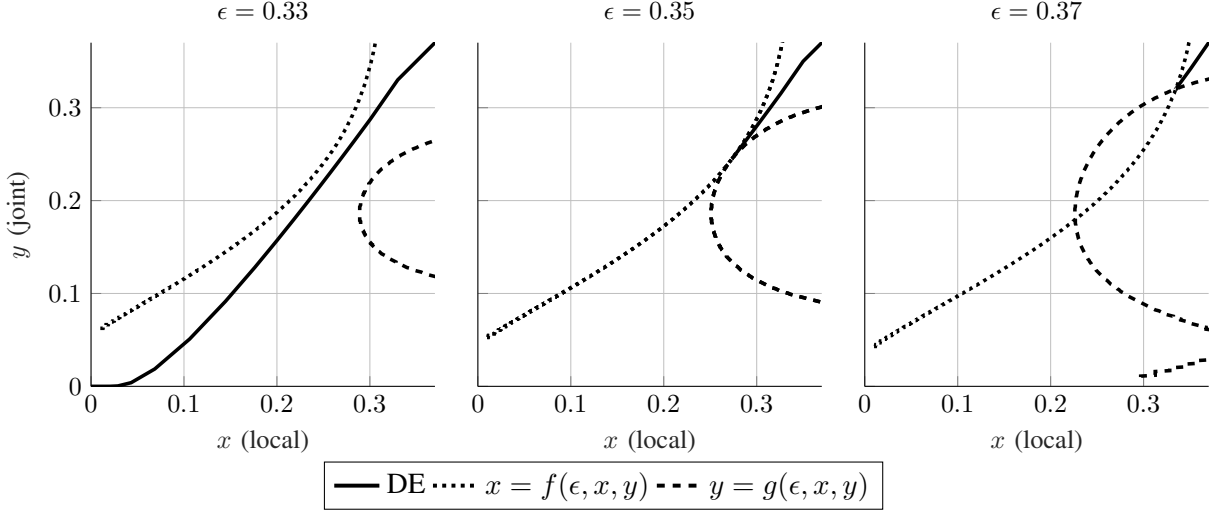}
\end{center}
\caption{\label{Fig:DE}
Illustration of the density-evolution equations in \eqref{eq:2D DE1}-\eqref{eq:2D DE3} for the LDPCL ensemble in Example~\ref{ex:threshold}, which induce a global decoding threshold of $\epsilon^*=0.35$. The evolved channel erasure probabilities, from left to right, are $\epsilon=0.33,0.35,0.37$.}
\end{figure}

\section{LDPCL Constructions and Achieving Capacity}
\label{Sec:C1}

In this section, we present an LDPCL ensemble construction, and show how to use this construction to optimally combine two degree distributions (local and joint) in order to approach capacity.
The inputs for the construction are the desired local and global decoding thresholds, $\epsilon_L$ and $\epsilon_G$, respectively, and the outputs are degree-distributions $\left(\lambda_L,\rho_L,\lambda_J,\rho_J,P_0 \right)$ such that
\begin{align*}
\epsilon_L^*\left(\lambda_L,\rho_L \right)=\epsilon_L \quad \epsilon^*_G\left(\lambda_L,\rho_L,\lambda_J,\rho_J,P_0 \right)=\epsilon_G.
\end{align*}
The specified parameters $\epsilon_L$ and $\epsilon_G$ can be arbitrarily chosen as fit for the specific application using the codes. $\epsilon_G$ is logically chosen to meet the "worst-case" noise level in extreme channel instances, while $\epsilon_L$ should specify a lower noise tolerance that is sufficient for a significant fraction of channel instances.  

In principle, setting $P_0=0$, and picking any two LDPC ensembles $\left(\lambda_L,\rho_L\right)$ and $\left(\lambda_J,\rho_J\right)$ that induce thresholds $\epsilon^*\left(\lambda_L,\rho_L\right)=\epsilon_L$ and $\epsilon^*\left(\lambda_J,\rho_J\right)=\epsilon_G$ would suffice, but this choice yields poor rates (intuitively, with that choice the local and joint codes do not "cooperate"). Another solution is not using a joint ensemble at all, i.e., choosing $\left(\lambda_L,\rho_L\right)$ such that $\epsilon^*\left(\lambda_L,\rho_L\right)=\epsilon_G > \epsilon_L$, and setting $P_0=1$. However, this solution is an undesired overkill since it would miss the opportunity to have a low-complexity local decoder for the majority of decoding instances where the erasure probabilities are below $\epsilon_L$.

\begin{definition} \label{def:h_eps}
Let $\left(\lambda_L,\rho_L \right)$ be local degree-distribution (DD) polynomials, and let $\epsilon^*\left(\lambda_L,\rho_L\right)=\epsilon_L$ be their decoding threshold.
For $\epsilon \in (\epsilon_L,1)$, let
\begin{enumerate}
\item $h_{\epsilon}(x)=\epsilon\lambda_L(1-\rho_L(1-x))-x,\quad x \in [0,1]$
\item \label{def:x_s}$x_s(\epsilon)=\max\{x \in [0,1] \colon h_\epsilon(x)\geq 0\}$ 
\item $a_s(\epsilon)=\Lambda_L\left(1- \rho_L\left(1-x_s(\epsilon) \right)\right)$
\end{enumerate} 
\end{definition}

For every $x \in [0,1]$, $h_\epsilon(x)$ is the erasure-probability change in one BP iteration on the local graph, if the current erasure probability is $x$. By definition, since $\epsilon>\epsilon_L$, $h_\epsilon(x)>0$ for some $x \in [0,1]$. In addition, for every $x>\epsilon$, $h_\epsilon(x) < 0$, so $x_s(\epsilon)$ is well defined. Operationally, $x_s(\epsilon)$ is the local-edge erasure probability when the local decoder gets stuck. Items 1 and 2 have appeared in \cite{KudRichUrb11}; we add $a_s(\epsilon)$ as a function of  $x_s(\epsilon)$ that encapsulates the erasure probability passed from the local code to the global code.

\begin{construction}\label{Construct:LDPCL_BEC}~

\underline{Input:} local threshold $\epsilon_{L}$ and global threshold $\epsilon_{G}>\epsilon_L$.

\begin{enumerate}
	\item \label{Item:local dd BEC}Choose any local DD $(\lambda_L,\rho_L)$ such that $\epsilon^*(\lambda_L,\rho_L)=\epsilon_{L}$.
	\item Calculate $ a_s(\epsilon_G) $.
	\item \label{item:eps_J} Choose any joint DD $(\lambda_J,\rho_J)$ such that $\epsilon^*(\lambda_J,\rho_J)=\epsilon_{G}\cdot a_s(\epsilon_G)$.
	\item Set $ P_0=\frac{\epsilon_L}{\epsilon_G} $.
\end{enumerate}	

\end{construction}

\begin{theorem}\label{th:C1}
Let $ (\lambda_L,\rho_L,\lambda_J,\rho_J,P_0)$ be LDPCL degree distributions constructed by Construction~\ref{Construct:LDPCL_BEC}. Then $\epsilon^*(\lambda_L,\rho_L)=\epsilon_L$ and $\epsilon^*(\lambda_L,\rho_L,\lambda_J,\rho_J,P_0)=\epsilon_G$.
\end{theorem}

\begin{proof}
The local threshold $ \epsilon_L $ follows trivially from item~\ref{Item:local dd BEC} above.

From Theorem~\ref{th:numerical th} and \eqref{eq:1D th numeric}, since $P_0=\frac{\epsilon_L}{\epsilon_G}>0$, we have
\begin{align}\label{eq:up bound th C1}
\epsilon^*_G \leq \frac1{P_0}\cdot \inf_{(0,1]}\frac{x}{\lambda_L\left(1- \rho_L\left(1-x \right)\right)} = \frac{\epsilon_L}{P_0} = \epsilon_G.
\end{align}

For the opposite direction, let $\epsilon<\epsilon_G$. In view of Theorem~\ref{theorem:fix pt char}, it suffices to show that \eqref{eq:f,g fixed} has no solution for $(x,y)\in (0,1]\times [0,1]$. In view of \eqref{eq:f}, for every $ x \in (0,1]$,
\begin{align} \label{eq:f<x y=0}
f({\epsilon},x,0) &= {\epsilon} \lambda_L(1-\rho_L(1-x))P_0 \notag \\
&< P_0 \epsilon_G \lambda_L(1-\rho_L(1-x)) \notag \\
&= \epsilon_L \lambda_L(1-\rho_L(1-x)) \notag \\
&\leq x .
\end{align}
Furthermore, Definition~\ref{def:h_eps} implies that for every $(x,y) \in (x_s(\bar{\epsilon}), 1)\times[0,1]$,  
\begin{align}\label{eq:f<x}
f({\epsilon},x,y) 
&= {\epsilon} \lambda_L(1-\rho_L(1-x)) \Lambda_J\left(1- \rho_J\left(1-y\right)\right) \notag \\
&< x\Lambda_J\left(1- \rho_J\left(1-y\right)\right) \notag \\
&\leq x,
\end{align}
and from Lemma~\ref{lemma:monotonicity}, if $(x,y) \in (0,x_s(\bar{\epsilon})]\times(0,1]$, 
\begin{align} \label{eq:<a_eps}
g({\epsilon},x,y)
&\leq g({\epsilon},x_s({\epsilon}),y) \notag \\
&= {\epsilon} \lambda_J(1-\rho_J(1-y)) \Lambda_L\left(1- \rho_L\left(1-x_s(\bar{\epsilon})\right)\right) \notag \\
&= {\epsilon} \lambda_J(1-\rho_J(1-y)) a_s({\epsilon}).
\end{align}
Since ${\epsilon}<\epsilon_G$, then $\epsilon_J=\epsilon_G \cdot a_s(\epsilon_G) > {\epsilon}\cdot a_s({\epsilon})$, where $\epsilon_J$ is the BP decoding threshold of $(\lambda_J,\rho_J)$; thus \eqref{eq:<a_eps} yields 
\begin{align} \label{eq:g<y}
g({\epsilon},x,y)< \epsilon_J \lambda_J(1-\rho_J(1-y)) \leq y ,\quad \forall (x,y) \in (0,x_s({\epsilon})]\times(0,1].
\end{align}
Combining \eqref{eq:f<x y=0}, \eqref{eq:f<x}, and \eqref{eq:g<y} implies that \eqref{eq:f,g fixed} has no solution in $(0,1]\times [0,1]$. Thus $\epsilon^*_G \geq \epsilon$. Since this is true for any ${\epsilon}<\epsilon_G$, we conclude that 
\begin{align*}
\epsilon^*_G \geq \epsilon_G,
\end{align*}
which combined with \eqref{eq:up bound th C1} completes the proof.
\end{proof}

\begin{remark}
In most cases, it is hard to produce an analytical expression for $x_s(\epsilon)$, but if we limit the local degrees of the ensemble to be small, then a closed-form expression could be derived for $x_s(\epsilon),a_s(\epsilon)$, and $\epsilon_J$.
\end{remark}

\begin{example}

Consider local ensembles taking the form:
\begin{align} \label{eq:low local2}
	\lambda_L(x)=x,\, \rho_L(x) = \rho_2x+\rho_3x^2+\rho_4x^3, \quad \rho_i\geq 0,\,\rho_2+\rho_3+\rho_4=1.
\end{align}  
In view of \eqref{eq:1D th numeric}, for the family of ensembles given in \eqref{eq:low local2}, $\epsilon_L=\frac1{1+\rho_3+2\rho_4}\;.$
In addition, for every $\epsilon \in (\epsilon_L,1)$,
\begin{align*}
	x_s(\epsilon)=\left\{
	\begin{array}{ll}
	1-\tfrac1{\rho_3}\left(\tfrac1\epsilon-1 \right), & \rho_4=0 \\
	\tfrac{\rho_3+3\rho_4-\sqrt{\left( \rho_3+\rho_4\right)^2+4\rho_4\left(\tfrac1\epsilon-1 \right)}}{2\rho_4}, & \rho_4>0
	\end{array}
	\right.\;.
\end{align*}
Finally, $a_s(\epsilon)=\left(\frac{x_s(\epsilon)}{\epsilon}\right)^2\;.$ 
These closed-form expressions of $ x_s $ and $ a_s $ can be used for constructing a code with certain parameters $\epsilon_L\;, \epsilon_G$, using a simple optimization of the parameters $\rho_2\;,\rho_3\;,\rho_4$ (we omit the details here).

\end{example}

\subsection{Achieving Capacity} \label{sub:c.a. seq}

In this sub-section, we define the LDPCL notion of BEC capacity-achieving sequences, and prove that such sequences exist. During the derivation, we refer to $\delta(\lambda,\rho)$ as the additive gap to capacity of the $LDPC(\lambda,\rho)$ ensemble, i.e., $\delta(\lambda,\rho)=1-\epsilon^*(\lambda,\rho)-R(\lambda,\rho)$. Similarly, we define  $\delta\left(\lambda_L,\rho_L,\lambda_J,\rho_J,P_0\right)=1-\epsilon^*_G(\lambda_L,\rho_L,\lambda_J,\rho_J,P_0)-R(\lambda_L,\rho_L,\lambda_J,\rho_J,P_0)$ as the global additive gap to capacity.

\begin{definition} \label{def:c.a. seq}
	Let $ 0 < \epsilon_L < \epsilon_G<1$. A sequence $\left\{ \lambda^{(k)}_L,\rho^{(k)}_L,\lambda^{(k)}_J,\rho^{(k)}_J,P^{(k)}_0\right\}_{k \geq 1}$ is said to achieve capacity on a $BEC(\epsilon_G)$, with local decoding capability $\epsilon_L$ if:
	\begin{enumerate}
		\item $\lim\limits_{k \to \infty}\epsilon^*_L\left(\lambda^{(k)}_L,\rho^{(k)}_L\right)=\epsilon_L$
		\item $\lim\limits_{k \to \infty}\epsilon^*_G\left(\lambda^{(k)}_L,\rho^{(k)}_L,\lambda^{(k)}_J,\rho^{(k)}_J,P^{(k)}_0\right)=\epsilon_G$ 
		\item $\lim\limits_{k \to \infty}R\left(\lambda^{(k)}_L,\rho^{(k)}_L,\lambda^{(k)}_J,\rho^{(k)}_J,P^{(k)}_0\right)=1-\epsilon_G$ 
	\end{enumerate}
\end{definition}
Note that items~2 and~3 imply that $\lim\limits_{k \to \infty}\delta\left(\lambda^{(k)}_L,\rho^{(k)}_L,\lambda^{(k)}_J,\rho^{(k)}_J,P^{(k)}_0\right)=0$.

\begin{lemma} \label{lemma:gap2cap}
	Let $\left(\lambda_L,\rho_L,\lambda_J,\rho_J,P_0\right)$ be degree-distribution polynomials constructed according to Construction~\ref{Construct:LDPCL_BEC}, and let $\delta_L \triangleq \delta(\lambda_L,\rho_L)$ and $\delta_J\triangleq\delta(\lambda_J,\rho_J)$. Then,
	\begin{align} \label{eq:delta}
	\delta\left(\lambda_L,\rho_L,\lambda_J,\rho_J,P_0\right) \leq \delta_L+\delta_J\cdot \left( 1-P_0\right).
	\end{align}
\end{lemma}

\begin{proof}
	Let $\epsilon_G=\epsilon^*_G\left(\lambda_L,\rho_L,\lambda_J,\rho_J,P_0\right)$ be the global threshold. In view of Theorem~\ref{th:C1}, $P_0=\tfrac{\epsilon_L}{\epsilon_G}$, and
	\begin{align} \label{eq:epsJ<eps}
	\epsilon_J \triangleq \epsilon^*\left( \lambda_J,\rho_J\right)= \epsilon_G \cdot a_s(\epsilon_G) \leq \epsilon_G.
	\end{align}
	In addition, by definition we have,
	\begin{align} \label{eq:deltaL deltaJ}
	\begin{split}
	\frac{\int_0^1 \rho_L(x)\mathrm{d}x}{\int_0^1 \lambda_L(x)\mathrm{d}x} = \epsilon_L+\delta_L ,\quad  \frac{\int_0^1 \rho_J(x)\mathrm{d}x}{\int_0^1 \lambda_J(x)\mathrm{d}x} = \epsilon_J+\delta_J,
	\end{split}
	\end{align}
	hence \eqref{eq:design rate2}, \eqref{eq:epsJ<eps} and \eqref{eq:deltaL deltaJ} imply,
	\begin{align}\label{eq:gap2cap proof}
	\delta\left(\lambda_L,\rho_L,\lambda_J,\rho_J,P_0\right) 
	&=  1-R\left(\lambda_L,\rho_L,\lambda_J,\rho_J,P_0\right)-\epsilon^*_G\left(\lambda_L,\rho_L,\lambda_J,\rho_J,P_0\right) \notag \\
	&=\epsilon_L+\delta_L+\left(\epsilon_J+\delta_J\right)\left( 1-\frac{\epsilon_L}{\epsilon_G}\right) -\epsilon_G \notag \\
	&=\frac1{\epsilon_G}\overbrace{(\epsilon_J-\epsilon_G)}^{\leq 0}\overbrace{(\epsilon_G-\epsilon_L)}^{\geq 0}+\delta_L+\delta_J\left( 1-\frac{\epsilon_L}{\epsilon_G}\right) \notag \\
	&\leq  \delta_L+\delta_J\left(1-\frac{\epsilon_L}{\epsilon_G}\right) \\
	&=  \delta_L+\delta_J\left(1-P_0\right).
	\end{align}
\end{proof}

\begin{remark}
	In principle, the bound in Lamma~\ref{lemma:gap2cap} may not be tight since $ \epsilon_J $ may be strictly smaller than $ \epsilon_G $. In this case the gap to capacity is smaller. However, due to dependencies between $ \epsilon_G $ and $\epsilon_J$ (see item~\ref{item:eps_J} in Construction~\ref{Construct:LDPCL_BEC}), we assume this worst case to simplify the bound. 
\end{remark}
At this point, it should be clear how to construct a capacity-achieving sequence of LDPCL ensembles on a $BEC(\epsilon_G)$, with a local decoding capability $\epsilon_L$. Choose any two sequences of (ordinary) LDPC ensembles $\left\{ \lambda^{(k)}_L,\rho^{(k)}_L,\right\}_{k \geq 1}$ and $\left\{ \lambda^{(k)}_J,\rho^{(k)}_J,\right\}_{k \geq 1}$ that achieve capacity on the $BEC(\epsilon_L)$ and $BEC(\epsilon_G)$, respectively, and set $P_0^{(k)}=\left(1-\frac{\epsilon_L}{\epsilon_G}\right)$, for all $k \geq 1$. Item 1 in Definition~\ref{def:c.a. seq} clearly holds for this  sequence, and in view of Theorem~\ref{th:C1}, item 2 in Definition~\ref{def:c.a. seq} holds as well. Finally, Lemma~\ref{lemma:gap2cap} implies that
\begin{align*}
\lim\limits_{k \to \infty}\delta\left(\lambda^{(k)}_L,\rho^{(k)}_L,\lambda^{(k)}_J,\rho^{(k)}_J,P_0\right)  
&\leq  \lim\limits_{k \to \infty}\delta\left(\lambda^{(k)}_L,\rho^{(k)}_L\right)  +  \lim\limits_{k \to \infty}\delta\left(\lambda^{(k)}_J,\rho^{(k)}_J\right)\left(1-\frac{\epsilon_L}{\epsilon}\right) =0.
\end{align*} 

\begin{example} \label{ex:c.a.}
	We construct an LDPCL capacity-achieving sequence with local and global threshold $\epsilon_L=0.05$ and $\epsilon_G=0.2$, respectively. We set $P_0=\tfrac{\epsilon_L}{\epsilon_G}=0.25$ and we use the Tornado capacity-achieving sequence \cite{Luby01},
	\begin{align} \label{eq:Tornado}
	\begin{split}
	& \lambda_L^{(D_L)}(x)=\frac1{H(D_L)}\sum_{i=1}^{D_L} \frac{x^i}{i}, \qquad
	\lambda_J^{(D_J)}(x)=\frac1{H(D_J)}\sum_{i=1}^{D_J} \frac{x^i}{i}, \\
	& \rho_L^{(D_L)}(x)=e^{-\alpha_L}\sum_{i=0}^{\infty} \frac{(\alpha_L x)^i}{i!}, \qquad
	\rho_J^{(D_J)}(x)=e^{-\alpha_J}\sum_{i=0}^{\infty} \frac{(\alpha_J x)^i}{i!},
	\end{split}
	\end{align}
	where $H(\cdot)$ is the harmonic sum, $\alpha_L=\tfrac{H(D_L)}{\epsilon_L}$ (the check degree-distribution series are truncated to get degree-distribution polynomials with finite degrees). $D_L$ (resp. $D_J$) controls the local (resp. joint) gap to capacity $\delta_L$ (resp. $\delta_J$); the bigger it is, the smaller the gap is. Table~\ref{tbl:tornado} exemplifies how the LDPCL sequence $\left\{\lambda_L^{(D_L)},\rho_L^{(D_L)},\lambda_J^{(D_J)},\rho_J^{(D_J)},P_0 \right\}$ approaches capacity as $D_L\to \infty,\;D_J\to \infty$: Theorem~\ref{th:C1} implies that for every value of $D_L$ and $D_J$, the global decoding threshold is $\epsilon^*_G\geq0.2$; the local additive gap to capacity $\delta_L$ and joint additive gap to capacity $\delta_J$ both vanish as $D_L\to \infty$ and $D_J\to \infty$, which in view of \eqref{eq:delta}, implies that the global additive gap to capacity $\delta$ vanishes as well.
	
	\begin{table}
		\caption{\label{tbl:tornado} Local gap to capacity $\delta_L$, joint gap to capacity $\delta_J$, and rate of the LDPCL capacity-achieving sequence given by \eqref{eq:Tornado} with $\epsilon_L=0.05$, $\epsilon_G=0.2$, and $P_0=0.25$.}
		\begin{center}
			
			\begin{tabular}{ c|c|c|c|c}
				\hline
				$D_L$&$\delta_L$&$D_J$&$\delta_J$&Rate    \\ \hline \hline
				1    &  0.05    &  1  &  0.2     & 0.6    \\
				1    &  0.05    &  2  &  0.1     & 0.67   \\
				1    &  0.05    &  10 &  0.02    & 0.735  \\
				1    &  0.05    &  100&  0.002   & 0.745  \\
				2    &  0.025   &  100&  0.002   & 0.775  \\
				5    &  0.01    &  100&  0.002   & 0.79   \\
				$\infty$&     0    &$\infty$&  0    & 0.8    
			\end{tabular}
		\end{center}
	\end{table}
	
\end{example}

\begin{remark}
	Table~\ref{tbl:tornado} shows the advantage of the multi-block scheme: one can get very close to capacity with local ensembles that are extremely low complexity thanks to their low $D_L$ values in the left column. 
\end{remark}

\begin{example}\label{Ex:BEC Sim}
	Figure~\ref{Fig:Sim Global BEC} compares the BEC global-decoding performance of an LDPCL code with ordinary LDPC codes. The $ M=4,n=2^{11} $ LDPCL code was constructed by Construction~\ref{Construct:LDPCL_BEC} with local degree distributions $ \lambda_L(x)=x^2,\;\rho_L(x)=x^{23} $ ($ \epsilon_L=0.1038 $), and capacity-achieving joint degree distributions from \eqref{eq:Tornado} with $D=30,\; \epsilon_G=0.45  $. The resulting code is of length $ 2^{13} $ with rate $ R=0.53 $. The two LDPC codes have the same rate and asymptotic threshold as the LDPCL code, and their lengths are $ n=2^{11} $ and $ n=2^{13} $ for the short and long block lengths, respectively. 	
	As seen in Figure~\ref{Fig:Sim Global BEC}, the LDPCL code outperforms in global decoding both the short and long ordinary LDPC codes. Improving over the short block is largely thanks to the increase in block length, and improving over the long code comes from the better finite block length behaviors of the local and joint degree distributions vs. the single degree-distribution pair of the LDPC code. 
	Moreover, the number of local edges in the Tanner graph of the LDPCL code is 6130 and the number of edges in the short and long block LDPC codes are 8365 and 33235, respectively.
	Consequently, choosing the LDPCL option allows both low-complexity decoding of small sub-blocks and increased reliability for the full block. 
	      
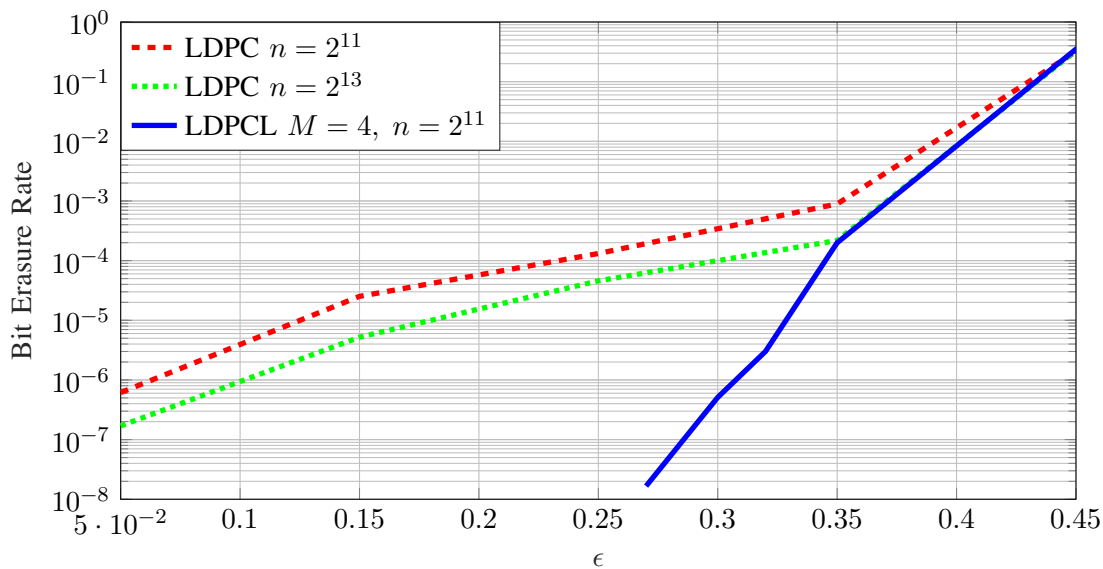
\begin{figure}[!h]
\begin{center}
	\begin{tikzpicture}

\begin{axis}[%
width=5in,
height=2.5in,
at={(0,0)},
scale only axis,
xmin=0.05,
xmax=0.45,
xlabel style={font=\color{white!15!black}},
xlabel={$\epsilon$},
ylabel style={font=\color{white!15!black}},
ylabel={Bit Erasure Rate},
ymode=log,
ymin=1e-08,
ymax=1,
yminorticks=true,
axis background/.style={fill=white},
xmajorgrids,
ymajorgrids,
yminorgrids,
legend style={at={(0,1)}, anchor=north west, legend cell align=left, align=left, draw=white!15!black}
]
\addplot [red,dashed,line width=2.0pt]
  table[row sep=crcr]{%
0.05	6.14742444841018e-07\\
0.15	2.51842059108836e-05\\
0.25	0.000131637360143128\\
0.35	0.000892787524366472\\
0.45	0.31893534762833\\
};
\addlegendentry{LDPC $ n=2^{11} $}

\addplot [green,dotted,line width=2.0pt]
  table[row sep=crcr]{%
0.05	1.69847145180381e-07\\
0.15	5.21229030226705e-06\\
0.25	4.583628122519e-05\\
0.35	0.000217934627752676\\
0.45	0.330143118817926\\
0.47	0.399126144828428\\
0.5	0.453926079700411\\
};
\addlegendentry{LDPC $ n=2^{13} $}

\addplot [blue,line width=2.0pt]
table[row sep=crcr]{%
	0.27	1.66439755585446e-08\\
	0.3	5.11108660321349e-07\\
	0.32	3.01186912488809e-06\\
	0.35	0.000199625651041667\\
	0.45	0.35078271484375\\
	0.47	0.413217041015625\\
	0.5	0.464005045572917\\
};
\addlegendentry{LDPCL $ M=4,\;n=2^{11}$}

\end{axis}
\end{tikzpicture}%
\end{center}
\caption{\label{Fig:Sim Global BEC} Global decoding performance over the BEC$( \epsilon )$. All codes are of rate $ R=0.53 $.  }
\end{figure}
\end{example}

\section{Reducing The Number of Joint Iterations}
\label{Sec:N_JI}
It has not been emphasized earlier in the paper, but in practical settings, the local and joint decoding iterations may be very different in terms of cost. Joint iterations access a much larger (factor $ M $) data unit, which may involve a high cost of transferring the bits to the check-node logic. In distributed-storage applications, the sub-blocks may even reside in remote sites, increasing the communication cost even further. Therefore, we would like to reduce the number of joint iterations (call it $N_{JI}$) performed by the decoder (i.e., rounds of variable-to-joint-check messages and joint-check-to-variable messages). Ideally, the decoder successfully decodes the desired sub-block (of length $n$) on the local graph (see Section~\ref{sub:Tanner}), does not enter the global mode, and no joint iterations are needed ($N_{JI}=0$); in the asymptotic regime, this happens when the fraction of erased bits is equal or less than the local threshold, i.e., $\epsilon \leq \epsilon_L$. However, if $\epsilon>\epsilon_L$, then the decoder fails to decode in the local mode and it enters the global mode where at least one joint iteration is necessary ($N_{JI}\geq 1$). 

In this section, we suggest a scheduling scheme for updating the joint side of the sub-blocked Tanner graph during global-mode decoding. We prove that our scheduling scheme is optimal in the sense of minimizing $N_{JI}$. In addition, we study how the parameters of the local and joint ensembles affect $N_{JI}$. Note that our notion of scheduling differs from the standard meaning of scheduling algorithms for iterative-decoding (see \cite{ZhangFossorier02, XiaoBeni04, CasGriotWesel07}). We consider scheduling of joint decoding iterations, while previous work considered the order of message passing between nodes in the Tanner graph.

\subsection{An $N_{JI}$-optimal scheduling scheme}
\label{sub:opt schedule}
Recall the LDPCL density-evolution equations:
\begin{subequations}
\begin{align}
\label{eq:2D-DE1}
&x_l = f\left(\epsilon,x_{l-1},y_{l-1}\right), \quad l \geq 0,\\
\label{eq:2D-DE2}
&y_l = g\left(\epsilon,x_{l-1},y_{l-1}\right),\quad l \geq 0,\\
\label{eq:2D-DE3}
&x_{-1}=y_{-1}=1,
\end{align}
\end{subequations}
where $f$ and $g$ are given in \eqref{eq:f} and \eqref{eq:g}, respectively; note that \eqref{eq:2D-DE1} and \eqref{eq:2D-DE2} express a local and a joint iteration, respectively. A scheduling scheme prescribes decoder access to the joint check nodes in only part of the iterations, and thus replaces \eqref{eq:2D-DE2} with 
\begin{align}\label{eq:schedule}
y_l = \left\{
\begin{array}{ll}
 g\left(\epsilon,x_{l-1},y_{l-1}\right) & l \in A \\
y_{l-1} & l \notin A
\end{array}
\right.
\end{align}
for some $A \subseteq \mathbb{N}$ representing the iteration numbers where joint checks are accessed; in this case we have
$ N_{JI}=\left| A \right|.$
Note that when the check-node update is skipped in the joint side there is no need to access the variable nodes outside the sub-block. 
Since Lemma~\ref{lemma:mono of DE} (monotonicity) still holds when \eqref{eq:2D-DE2} is replaced with \eqref{eq:schedule}, the limits $\lim\limits_{l \to \infty}x_l$ and $\lim\limits_{l \to \infty}y_l$ exist for every scheduling scheme.

Given local and joint degree-distributions, a scheduling scheme is called valid if for every $\epsilon< \epsilon_G=\epsilon^*_G(\Lambda_L,\Lambda_J,\Omega_L,\Omega_J)$, $\lim\limits_{l \to \infty}x_l(\epsilon)=0$ (successful decoding). Our goal is to find an optimal scheduling scheme: a valid scheduling scheme that minimizes $N_{JI}$.  For example, if $A=\emptyset$ (no joint updates), then $N_{JI}=0$ but $\lim\limits_{l \to \infty}x_l(\epsilon)>0$ if $\epsilon \in (\epsilon_L, \epsilon_G)$; thus, the scheduling scheme is not valid. If, on the other hand, joint checks are accessed in every iteration (as assumed in Sections~\ref{Sec:2D DE}--\ref{Sec:C1}, then the scheduling scheme is valid, but $N_{JI}$ equals the total number of iterations, which is the worst case. We do not require the scheduling scheme to be pre-determined, and it can use ``on-line" information about the decoding process. For example, it can use the current fraction of erasure messages or the change in this fraction between two consecutive iterations. 
\begin{definition} \label{def:eff epsilon}
Let $(\Lambda_J,\rho_J)$ be joint degree-distribution polynomials, let $\epsilon\in (0,1)$ be the erasure probability of a $BEC$, and let $y \in [0,1]$ be an instantaneous erasure probability from the joint perspective. We define the effective erasure probability from the local perspective as
\begin{align} \label{eq:eff loc eps}
\epsilon_{loc}(y) = \epsilon \cdot \Lambda_J\left(1-\rho_J(1-y)\right).
\end{align}
\end{definition}
In view of \eqref{eq:f} and \eqref{eq:eff loc eps}, we have
\begin{align} \label{eq:f with eff}
x_l = f(\epsilon,x_{l-1},y_{l-1})=\epsilon_{loc}(y_{l-1}) \lambda_L(1-&\rho_L(1-x_{l-1})).
\end{align}
$\epsilon_{loc}(y_{l-1}) $ takes the role of $\epsilon$ when the local code is viewed as a standard LDPC code, hence the term "effective erasure probability from the local perspective".

Our proposed scheduling scheme is parametrized by $\eta>0$, and is given by
\begin{align}\label{eq:opt schedule}
\begin{split}
&x_l=f(\epsilon,x_{l-1},y_{l-1}) \\
&y_l = \left\{
\begin{array}{ll}
g\left(\epsilon,x_{l-1},y_{l-1}\right) & \left| x_{l-2}-x_{l-1}\right|\leq \eta \text{ and } \epsilon_{loc}(y_{l-1})\geq \epsilon_L \\
y_{l-1} & \text{else}
\end{array}
\right. .
\end{split} 
\end{align}
\begin{lemma} \label{lemma:valid opt schedule}
For every $\eta>0$, the scheduling scheme described in \eqref{eq:opt schedule} is valid.
\end{lemma}
\begin{proof}
See Appendix~\ref{App:valid opt schedule}.
\end{proof}
Note that if $\eta=0$, the scheduling scheme described in \eqref{eq:opt schedule} is not valid. However, since local iterations have zero cost in our model, we can assume that we can apply arbitrarily many local iterations to get arbitrarily close to $\eta=0$. Numerical simulations show that $\eta=10^{-4}$ suffices for achieving minimal $N_{JI}$. For the following analysis we will assume that $\eta=0$, and that the scheduling scheme is still valid. In this scheduling scheme, the decoder tries to decode the sub-block on the local graph until it gets ``stuck", which refers to not being able to reduce the erasure probability while it is still strictly greater than zero. This happens first when $x_{l_1}=x_s(\epsilon)$, for some iteration $l_1$, where $x_s(\epsilon)$ is given in Definition~\ref{def:h_eps}. So, in the first joint update we have
\begin{align*}
\begin{split}
x_{l_1} &= x_s(\epsilon) ,\quad y_{l_1} =1 \\
x_{l_1+1} &= x_s(\epsilon) ,\quad y_{l_1+1} =g(\epsilon,x_s(\epsilon),1). 
\end{split}
\end{align*}
In view of \eqref{eq:f with eff}, the local graph now "sees" $\epsilon_{loc}(y_{l_1+1})<\epsilon$ as an effective erasure probability, and it can continue the decoding algorithm locally. It may get "stuck" again and another joint update will be invoked; this procedure continues until $\epsilon_{loc}(y_{l_p+1}) < \epsilon_L$ in the $p$'th (and last) update, which enables successful local decoding (i.e., $N_{JI}=p$). In general, let $\{l_k\}_{k=1}^{N_{JI}}$ be the joint update iterations of the scheduling scheme described above and let $\varepsilon_k $ be the effective erasure probability from the local perspective between joint updates $k-1$ and $k$. Then,  
\begin{subequations}
\begin{align}
\label{eq:y_1}
&y_{l_1}=1,\;\varepsilon_1=\epsilon,\;x_{l_1}=x_s(\epsilon), \\
\label{eq:y_l_k}
&y_{l_k}=g\left (\epsilon,x_{l_{k-1}},y_{l_{k-1}}\right ),\quad 2\leq k \leq N_{JI}, \\
\label{eq:eps_k}
&\varepsilon_k=\epsilon_{loc}(y_{l_k}),\quad 2\leq k \leq N_{JI}, \\
\label{eq:x_l_k}
&x_{l_k}=x_s(\varepsilon_k),\quad 2\leq k \leq N_{JI},
\end{align}
\end{subequations}
where
\begin{align} \label{eq:eps_k mono}
\begin{split}
\epsilon=\varepsilon_1 > \varepsilon_2 > \ldots > \varepsilon_{N_{JI}-1} \geq  \epsilon_L >  \varepsilon_{N_{JI}}.
\end{split}
\end{align}
\begin{lemma} \label{lemma:opt schedule}
The scheduling scheme described above is optimal. 
\end{lemma}
\begin{proof}
See Appendix~\ref{App:opt schedule}.
\end{proof}
We assume from now on that the decoder applies the optimal scheduling scheme.

\subsection{The Rate vs. $N_{JI}$ Trade-Off}

As shown in Section~\ref{sub:c.a. seq}, for every  $0<\epsilon_L<\epsilon<1$, there exist capacity-achieving LDPCL sequences for the $BEC(\epsilon)$ with a local decoding threshold of $\epsilon_L$. Moreover, Theorem~\ref{th:C1} and Lemma~\ref{lemma:gap2cap} provide a construction of such a sequence. However, we will now see that the closer an LDPCL ensemble is to capacity, the higher the $N_{JI}$ is; therefore, to decrease $N_{JI}$ we have to pay with rate, and there are several ways to do so. In this section we study how the parameters of the local and joint degree-distributions, $\lambda_L,\rho_L,\lambda_J,\rho_J,P_0$, affect $N_{JI}$. In particular, we focus on how the local and joint additive gaps to capacity $\delta_L$ and $\delta_J$, receptively, affect $N_{JI}$.  

It is well known that if $\left\{ \lambda^{(k)},\rho^{(k)}\right\}^{\infty}_{k=1}$ is a (ordinary) capacity-achieving sequence for the $BEC(\epsilon)$, then 
\begin{align}\label{eq:c.a. known}
\lim_{k \to \infty} \epsilon\lambda^{(k)}(1-\rho^{(k)}(1-x))=x,\quad x \in [0,\epsilon]
\end{align}
(see \cite{Shok01}).
This leads to the following lemma.
\begin{lemma} \label{lemma:x_s c.a.}
Let  $0<\epsilon_L<\epsilon<1$, and let $\left\{ \lambda_L^{(k)},\rho_L^{(k)}\right\}^{\infty}_{k=1}$ be a capacity-achieving sequence for the $BEC(\epsilon_L)$. Then,
\begin{align*}
x_s(\epsilon)\triangleq\lim_{k \to \infty}x^{(k)}_s(\epsilon)=\epsilon,
\end{align*}
where $x^{(k)}_s(\epsilon)$ corresponds to Definition~\ref{def:h_eps} with $\left(\lambda_L^{(k)},\rho_L^{(k)}\right)$.
\end{lemma}

\begin{proof}
See Appendix~\ref{App:x_s c.a.}.
\end{proof}

Lemma~\ref{lemma:x_s c.a.} asserts that if the \emph{local} degree-distribution polynomials imply a local threshold $\epsilon_L$ and a design rate that is very close to capacity ($1-\epsilon_L$), and the channel erasure probability $\epsilon$ is grater than $\epsilon_L$, then the BP decoding algorithm on the local graph gets ``stuck" immediately after correcting only a small fraction of the erasures. This leads, in view of \eqref{eq:y_l_k}, to a small change in the erasure-message probability on the joint update, which in turn yields a minor progress on the local side. Therefore, choosing close to capacity \emph{local} degree-distribution polynomials implies high $N_{JI}$.
Another consequence of \eqref{eq:c.a. known} is that the change in the erasure-message probability in one iteration of the BP decoding algorithm is small. Thus, close to capacity \emph{joint} degree-distribution polynomials yield high $N_{JI}$, regardless of the \emph{local} degree-distribution polynomials.

\begin{example} \label{ex:N_JI}
Let $\epsilon_L=0.05$ and $\epsilon_G=0.2$. We use the capacity-achieving LDPCL sequence for the $BEC(\epsilon_G)$ with a local decoding threshold of $\epsilon_L$ given in Example~\ref{ex:c.a.}. 
A computer program simulated \eqref{eq:y_1}-\eqref{eq:x_l_k} with \eqref{eq:Tornado} for different values of $D_L$ and $D_J$, and the results are presented in Figure~\ref{Fig:EL05P1EG2}. The plot exemplifies the trade-off between rate and $N_{JI}$: when the ensemble is close to capacity with $\delta_L=10^{-2},\delta_J=2.5\cdot10^{-4}$ ($R=0.79$), we get $N_{JI}=570$, and to reduce $N_{JI}$ we have to pay with rate. However, there are several ways to do so. For example, changing the local gap to $\delta_L=5\cdot10^{-2}$ while the joint gap stays $\delta_J=2.5\cdot10^{-4}$ yields $R=0.75$ and $N_{JI}=26$, and changing the local and joint gap to $\delta_L=2.5\cdot10^{-2}$ and $\delta_J=4\cdot10^{-2}$, respectively, yields the same $R=0.75$ but a smaller $N_{JI}=11$.

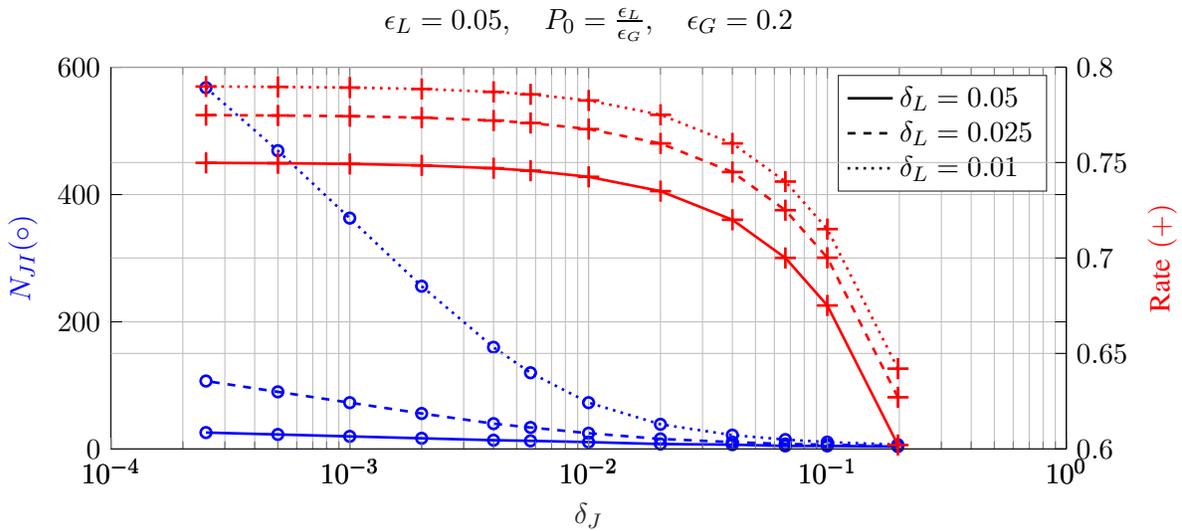
\begin{figure}[!h]
\begin{center}
\begin{tikzpicture}

\begin{axis}[%
width=5in,
height=2in,
scale only axis,
xmode=log,
xmin=0.0001,
xmax=1,
xminorticks=true,
xlabel style={font=\color{white!15!black}},
xlabel={$\delta_J$},
separate axis lines,
every outer y axis line/.append style={black},
every y tick label/.append style={font=\color{black}},
every y tick/.append style={black},
ymin=0,
ymax=600,
ylabel={\textcolor{blue}{$N_{JI} (\circ)$}},
axis background/.style={fill=white},
title={$\epsilon_L=0.05,\quad P_0=\frac{\epsilon_L}{\epsilon_G},\quad \epsilon_G=0.2$},
xmajorgrids,
xminorgrids,
ymajorgrids,
legend entries={$\delta_L=0.05$,$\delta_L=0.025$,$\delta_L=0.01$},
legend style={legend cell align=left, align=left, draw=white!15!black}
]
\addlegendimage{no markers,black,solid, line width=1.0pt}
\addlegendimage{no markers,black,dashed, line width=1.0pt}
\addlegendimage{no markers,black,dotted, line width=1.0pt}
\addplot [color=blue, line width=1.0pt,mark=o, mark options={solid}]
  table[row sep=crcr]{%
0.197254817033647	4\\
0.0998340372099368	5\\
0.0666388095230959	5\\
0.0399973574944374	7\\
0.0199999040049642	8\\
0.0099999967653871	11\\
0.00571428551077802	13\\
0.0039999999653707	14\\
0.00199999999893696	17\\
0.000999999999992784	20\\
0.00050000000003414	23\\
0.000250000000009409	26\\
};

\addplot [color=blue, dashed, line width=1.0pt,mark=o, mark options={solid}]
  table[row sep=crcr]{%
0.197254817033647	5\\
0.0998340372099368	7\\
0.0666388095230959	8\\
0.0399973574944374	11\\
0.0199999040049642	16\\
0.0099999967653871	25\\
0.00571428551077802	34\\
0.0039999999653707	40\\
0.00199999999893696	56\\
0.000999999999992784	73\\
0.00050000000003414	90\\
0.000250000000009409	107\\
};

\addplot [color=blue, dotted, line width=1.0pt,mark=o, mark options={solid}]
  table[row sep=crcr]{%
0.197254817033647	7\\
0.0998340372099368	11\\
0.0666388095230959	15\\
0.0399973574944374	22\\
0.0199999040049642	39\\
0.0099999967653871	73\\
0.00571428551077802	120\\
0.0039999999653707	160\\
0.00199999999893696	256\\
0.000999999999992784	363\\
0.00050000000003414	469\\
0.000250000000009409	568\\
};

\end{axis}

\begin{axis}[%
width=5in,
height=2in,
scale only axis,
xmode=log,
xmin=0.0001,
xmax=1,
xminorticks=true,
xlabel style={font=\color{white!15!black}},
xlabel={},
every outer y axis line/.append style={black},
every y tick label/.append style={font=\color{black}},
every y tick/.append style={black},
ymin=0.6,
ymax=0.8,
ylabel={\textcolor{red}{Rate $(+)$}},
axis x line*=bottom,
axis y line*=right,
ymajorgrids
]

\addplot [color=red, line width=1.0pt, mark=+,forget plot, mark options={solid},mark size=4pt]
  table[row sep=crcr]{%
0.197254817033647	0.602021384305841\\
0.0998340372099368	0.675124444189399\\
0.0666388095230959	0.700020893034882\\
0.0399973574944374	0.720001982085271\\
0.0199999040049642	0.735000072202418\\
0.0099999967653871	0.7425000026321\\
0.00571428551077802	0.745714286073054\\
0.0039999999653707	0.747000000232104\\
0.00199999999893696	0.74850000020694\\
0.000999999999992784	0.749250000206142\\
0.00050000000003414	0.749625000206116\\
0.000250000000009409	0.749812500206115\\
};
\addplot [color=red, line width=1.0pt, dashed, mark=+,forget plot, mark options={solid},mark size=4pt]
  table[row sep=crcr]{%
0.197254817033647	0.627021384099733\\
0.0998340372099368	0.70012444398329\\
0.0666388095230959	0.725020892828773\\
0.0399973574944374	0.745001981879163\\
0.0199999040049642	0.76000007199631\\
0.0099999967653871	0.767500002425992\\
0.00571428551077802	0.770714285866945\\
0.0039999999653707	0.772000000025996\\
0.00199999999893696	0.773500000000832\\
0.000999999999992784	0.774250000000033\\
0.00050000000003414	0.774625000000008\\
0.000250000000009409	0.774812500000007\\
};
\addplot [color=red, line width=1.0pt, dotted, mark=+,forget plot, mark options={solid},mark size=4pt]
  table[row sep=crcr]{%
0.197254817033647	0.642021384099757\\
0.0998340372099368	0.715124443983314\\
0.0666388095230959	0.740020892828797\\
0.0399973574944374	0.760001981879187\\
0.0199999040049642	0.775000071996334\\
0.0099999967653871	0.782500002426016\\
0.00571428551077802	0.785714285866969\\
0.0039999999653707	0.78700000002602\\
0.00199999999893696	0.788500000000856\\
0.000999999999992784	0.789250000000057\\
0.00050000000003414	0.789625000000032\\
0.000250000000009409	0.789812500000031\\
};
\end{axis}
\end{tikzpicture}%
\caption{\label{Fig:EL05P1EG2}
Plot of $N_{JI}$ (blue $\circ$ marks) and the design rate (red $+$ marks) as a function of the joint additive gap to capacity $\delta_J$ for different values of the local additive gap to capacity $\delta_L$.}
\end{center}
\end{figure}
\end{example}

\section{Finite-Length Analysis For ML Decoding}
\label{Sec:Finite}
In this section, we extend to LDPCL codes the finite block-length analysis of \cite[Lemma~B.2]{Di02}. We derive upper bounds on the expected block erasure probability under ML decoding of regular LDPCL ensembles. As shown in \cite{Di02}, the union bound is fairly tight for not too small block lengths, and the performance loss of iterative decoding compared to ML decoding is not too high. This motivates a bounding approach that simplifies the expressions by employing the union bound and the ML decoding analysis. The following bounds exemplify the effect on global decoding performance (in this case ML) caused by enabling local sub-block decoding.

For a sub-blocked Tanner graph $\mathcal{G}$ picked from the $\left(M,n,l_L,r_L,l_J,r_J\right)$-regular ensemble (i.e., $ \lambda_L(x)=x^{l_L-1} $, etc.), let $P_B^{ML}\left(\mathcal{G},\epsilon\right)$ designate the block-decoding-failure probability of the code $\mathcal{G}$ over the $BEC(\epsilon)$ when decoded by the ML decoder. In what follows, for a power series $f(x)=\sum_if_ix^i$, let $\mathrm{coef}\left( f(x),x^i\right)$ designate the coefficient of $x^i$ in $f(x)$, i.e., $\mathrm{coef}\left( f(x),x^i\right)=f_i$. 
\begin{definition}
For every $l,r,n,w \in \mathbb{N}$ such that $n\tfrac{l}{r}\in \mathbb{N}$, $l,r,n \geq 1$ and $w \leq n$, let
\begin{align}\label{eq:A(l,r,n,w)}
A(l,r,n,w)=\frac{\mathrm{coef}\left( \left( \frac{\left(1+x \right)^r+\left(1-x \right)^r}{2} \right)^{n\tfrac{l}{r}},x^{wl} \right)}{\binom{nl}{wl}}.
\end{align}
\end{definition}

\begin{theorem}[Union Bound - ML Decoding of regular LDPCL Codes]
~\begin{enumerate}
\item For $\epsilon \in (0,1)$, let $\bar{\epsilon}=1-\epsilon$.
\item For $n\in \mathbb{N}^+$, let $[n]=\{1,2,\ldots,n\}$, and ${[n]}_0=\{0\} \cup [n]$.
\item For $m \in \mathbb{N}$, let $e_1^m=\left( e_1,e_2,\ldots,e_m\right)$, where for every $i \in \{1,2,\ldots,m\}$, $e_i\in \mathbb{N}^+$, and let ${[e_1^m]}_0={[e_1]}_0 \times {[e_2]}_0\times\ldots\times {[e_m]}_0$.
\end{enumerate}
Then,
\begin{align}\label{eq:MLUB}
\begin{split}
\mathbb{E}\left[ P_B^{ML}\left(\mathcal{G},\epsilon\right)\right] \leq 
&\sum_{m=1}^M \binom{M}{m}\bar{\epsilon}^{(M-m)n} \sum_{e_1^m \in {[n]}^m} \prod_{i=1}^m \binom{n}{e_i} \epsilon^{e_i} \bar{\epsilon}^{(n-e_i)} \\
& \hspace*{0.2cm} \cdot\min\left\{ 1,-1+\sum_{w_1^m \in {[e_1^m]}_0 } A(l_J,r_J,Mn,w_1+w_2+\ldots+w_m)\prod_{i=1}^m A(l_L,r_L,n,w_i) \right\},
\end{split}
\end{align}
where the expectation is w.r.t. the $\left(M,n,l_L,r_L,l_J,r_J\right)$-regular ensemble.
\end{theorem}

\begin{proof}
Let $\mathcal{E} \subseteq \{1,2,\ldots,Mn\}$ be an erasure pattern specifying the indices of the erased code bits, and for every $i \in \{1,2,\ldots,M\}$, let $\mathcal{E}_i=\mathcal{E} \cap \{ 1+(i-1)n,2+(i-1)n,\ldots,n+(i-1)n\}$ be the erasure pattern in the $i$-th sub-block.
Consider the case that there are erasures only in the first $m$ sub-blocks, i.e., 
\begin{align} \label{eq:E_i>0}
\begin{split}
&|\mathcal{E}_i| >0 , \quad 1\leq i \leq m \\
&| \mathcal{E}_i| =0 , \quad l< m \leq M.
\end{split}
\end{align}
It is well known that the ML decoder fails if and only if $\mathrm{rank}\left( H_{\mathcal{E}}\right) < |\mathcal{E}|$, where $H$ is the parity-check matrix corresponding to a random member of the $(l_L,r_L,l_J,r_J)$-regular LDPCL ensemble, and $H_\mathcal{E}$ is the sub-matrix comprising the columns of $H$ indexed by $\mathcal{E}$. Applying the union bound yields
\begin{subequations}
\begin{align} 
\label{eq:UB1}
\Pr\left\{ \mathrm{rank}\left( H_{\mathcal{E}}\right) < |\mathcal{E}| \right\} 
&= \Pr\left\{ \exists x \in \mathbb{F}_2^{|\mathcal{E}|}\setminus \{0\},H_{\mathcal{E}}x^T=0\right\} \\
\label{eq:UB2}
&\leq \sum_{x \in \mathbb{F}_2^{|\mathcal{E}|}\setminus \{0\}}\Pr\left\{ H_{\mathcal{E}}x^T=0 \right\} \\
\label{eq:UB3}
&= -1+\sum_{x \in \mathbb{F}_2^{|\mathcal{E}|}}\Pr\left\{ H_{\mathcal{E}}x^T=0 \right\} \\
\begin{split}\label{eq:UB4}
&= -1+\sum_{x_1 \in \mathbb{F}_2^{|\mathcal{E}_1|}}\ldots\sum_{x_m \in \mathbb{F}_2^{|\mathcal{E}_m|}}\Pr\left\{ {H_0}_{\mathcal{E}}(x_1,\ldots,x_l)^T=0 \right\} \\ 
& \hspace*{4.5cm}\cdot\prod_{i=1}^m\Pr\left\{ {H_i}_{\mathcal{E}_i}x_i^T=0 \right\} 
\end{split}
\end{align}
\end{subequations}
The steps until getting to \eqref{eq:UB3} are the standard ones from \cite{Di02}, while \eqref{eq:UB4} follows from the sub-block structure introduced in LDPCL codes.
By marking $w_i=w(x_i)$ for every $i\in \{1,2,\ldots,m\}$, as the Hamming weight of the sub-word $x_i\in \mathbb{F}_2^{|\mathcal{E}_i|}$ we get (see \cite[Lemma B.2.]{Di02})
\begin{align} \label{eq:using A}
\begin{split}
&\Pr\left\{ {H_0}_{\mathcal{E}}(x_1,x_2,\ldots,x_m)^T=0 \right\} = A(l_J,r_J,Mn,w_1+w_2+\ldots +w_m). \\
&\Pr\left\{ {H_i}_{\mathcal{E}_i}x_i^T=0 \right\} = A(l_L,r_L,n,w_i),\quad 1\leq i \leq m.
\end{split}
\end{align}
Combining \eqref{eq:UB4} and \eqref{eq:using A} implies that if \eqref{eq:E_i>0} holds, then
\begin{align} \label{eq:UB+A}
\begin{split}
\Pr&\left\{ \mathrm{rank}\left(  H_{\mathcal{E}}\right) < |\mathcal{E}| \right\} \leq \\
&-1+\sum_{w_1 =0}^{|\mathcal{E}_1|}\sum_{w_2 =0}^{|\mathcal{E}_2|}\ldots\sum_{w_m =0}^{|\mathcal{E}_m|}A(l_J,r_J,Mn,w_1+w_2+\ldots +w_m) \prod_{i=1}^m A(l_L,r_L,n,w_i).
\end{split}
\end{align}

Since the $BEC$ is memoryless, then
\begin{align}\label{eq:prob EPS}
\Pr\left\{ \mathcal{E}\right\}=\bar{\epsilon}^{(M-m)n} \prod_{i=1}^m \epsilon^{|\mathcal{E}_i|} \bar{\epsilon}^{( n-|\mathcal{E}_i|)}.
\end{align}
From symmetry, \eqref{eq:UB+A} holds for every error pattern $\mathcal{E}$ that lies in exactly $m$ sub-blocks (i.e., not only the first $m$ sub-blocks). Finally, summing over $m \in \{1,2,\ldots,M\}$ and counting every error pattern imply
\begin{align} \label{fin proof}
\begin{split}
\mathbb{E}\left[ P_B^{ML}\left(\mathcal{G},\epsilon\right)\right] 
&=\sum_\mathcal{E}\Pr\left\{\mathcal{E}\right\}\Pr\left\{ \mathrm{rank}\left(  H_{\mathcal{E}}\right) < |\mathcal{E}| \right\} \\
&\leq\sum_{m=1}^M \binom{M}{m}\bar{\epsilon}^{(M-m)n} \sum_{e_1^m \in {[n]}^m} \prod_{i=1}^m \binom{n}{e_i} \epsilon^{e_i} \bar{\epsilon}^{(n-e_i)} \\
& \hspace*{0.2cm} \cdot\min\left\{ 1,-1+\sum_{w_1^m \in {[e_1^m]}_0 } A(l_J,r_J,Mn,w_1+w_2+\ldots+w_l)\prod_{i=1}^m A(l_l,r_L,n,w_i) \right\}.
\end{split}
\end{align}
\end{proof}

Figure~\ref{Fig:n120} plots the upper bound in \eqref{eq:MLUB} for the $(2,6,1,6)$-regular LDPCL ensemble, with localities $M=2,3$ and sub-block lengths $n=180,120$, respectively. Also shown is the upper bound in \cite[Lemma B.2]{Di02} for the standard $(3,6)$-regular LDPC ensemble. The ensembles in Figure~\ref{Fig:n120} have the same total block length of $360$ (in view of Remark~\ref{remark:asymmetry}, the iterative decoding performance of the $(2,6,1,6)$-regular LDPCL ensemble and the $(3,6)$-regular LDPC ensemble coincide as $n \to \infty$). Figure~\ref{Fig:n120} exemplifies the trade-off between the sub-block local-access capability ($M$) and the error correcting capability over the global block; for every $\epsilon \in [0.05,0.5]$ the block decoding failure probability increases as $M$ increases. While Figure~\ref{Fig:n120} is just an example with not very realistic parameters\footnote{Due to the combinatorial nature of the expression in (59), it is computationally difficult to evaluate the bounds for longer codes.}, we expect the same trade-off to apply in more generality: finer sub-block access with efficient local decoding has a cost in terms of the global-decoding performance (for the same rate and global code length).   

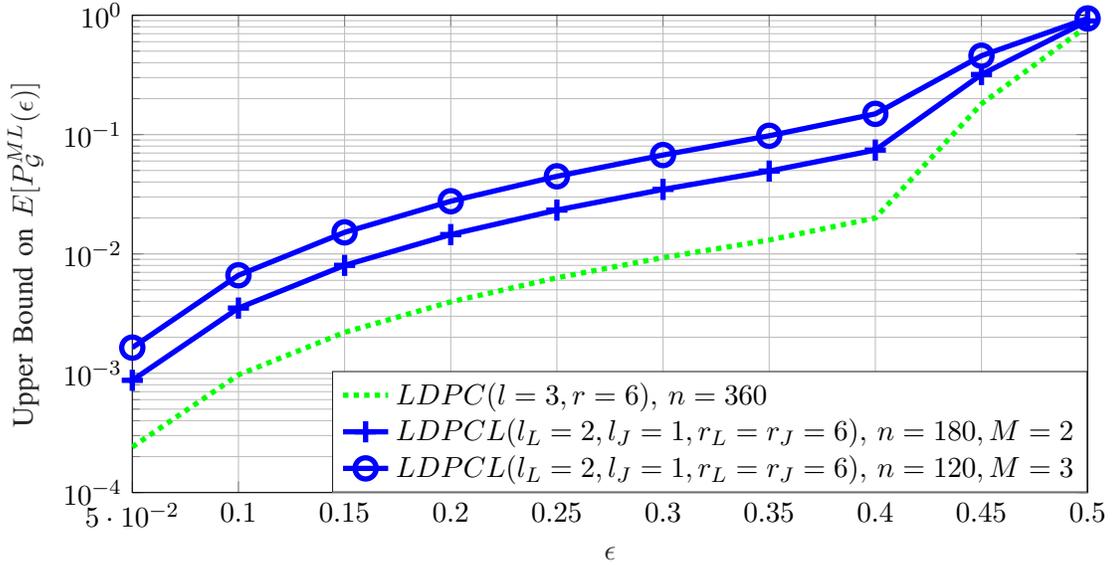
\begin{figure}[h!]
\begin{center}
\begin{tikzpicture}

\begin{axis}[%
width=5in,
height=2.5in,
scale only axis,
xmin=0.05,
xmax=0.5,
xlabel style={font=\color{white!15!black}},
xlabel={$\epsilon$},
ymode=log,
ymin=0.0001,
ymax=1,
yminorticks=true,
ylabel style={font=\color{white!15!black}},
ylabel={Upper Bound on $E[P_{\mathcal{G}}^{ML}(\epsilon)]$},
axis background/.style={fill=white},
xmajorgrids,
ymajorgrids,
yminorgrids,
legend style={at={(1,0)}, anchor=south east, legend cell align=left, align=left, draw=white!15!black}
]
\addplot [green,dotted, line width=2.0pt]
  table[row sep=crcr]{%
0.05	0.000240752347236284\\
0.1	0.000968322764120333\\
0.15	0.00219922183813419\\
0.2	0.0039630484975885\\
0.25	0.00630637847306853\\
0.3	0.0093004904375665\\
0.35	0.0130612365381501\\
0.4	0.0199891421804363\\
0.45	0.180841504559159\\
0.5	0.813804337871204\\
};
\addlegendentry{$LDPC(l=3,r=6),\,n=360$}

\addplot [blue,line width=2.0pt,mark=+, mark options={solid},mark size=4pt]
  table[row sep=crcr]{%
0.05	0.000871680424390182\\
0.1	0.00351366997122699\\
0.15	0.00801032053081061\\
0.2	0.0145151640277963\\
0.25	0.0232757684020197\\
0.3	0.0346851792515059\\
0.35	0.049416148369853\\
0.4	0.0741226714884994\\
0.45	0.318683744262964\\
0.5	0.894983933130287\\
};
\addlegendentry{$LDPCL(l_L=2,l_J=1,r_L=r_J=6),\,n=180,M=2$}

\addplot  [blue,line width=2.0pt,mark=o, mark options={solid},mark size=4pt]
  table[row sep=crcr]{%
0.05	0.00163800624900305\\
0.1	0.00661841133060244\\
0.15	0.015150631481671\\
0.2	0.027622039512717\\
0.25	0.0446738279663716\\
0.3	0.0673654593908212\\
0.35	0.0976313821610641\\
0.4	0.149071385663121\\
0.45	0.457512336840889\\
0.5	0.93970660025886\\
};
\addlegendentry{$LDPCL(l_L=2,l_J=1,r_L=r_J=6),\,n=120,M=3$}

\end{axis}
\end{tikzpicture}%
\caption{\label{Fig:n120}
Comparison of the bound in \eqref{eq:MLUB} ((2,6,1,6)-LDPCL code) with $M=2,3$ and $n=180,120$, respectively, with the bound in \cite[Lemma B.2]{Di02} ((3,6)-LDPC code) with $n=360$.}
\end{center}
\end{figure}

\section{Summary}
\label{Sec:sum}
This paper lays down the theoretical foundation for sub-block LDPC codes. This coding scheme enables fast read access to small blocks, and provides high data-protection on large blocks in case of more severe error events. 
We derived a code-analysis framework which resulted in a simple way to construct capacity-achieving sequences with guaranteed local-decoding performance. We found that the fraction of jointly-unconnected variable nodes, $P_0$, plays an important role in determining the asymptotic decoding threshold and in achieving capacity. 
It is known \cite{SasonUrba03} that there is an inherent trade-off between the gap to capacity and the encoding and decoding complexity of irregular LDPC codes over the BEC. Since it is of interest that the decoding algorithm will avoid joint-side messages as much as possible, we studied another trade-off regarding the number of joint iterations and the gap to capacity (see Example~\ref{ex:N_JI}). 
By deriving an upper bound on the block-erasure probability of finite-length LDPCL codes under ML decoding, we showed how the loss in performance is affected by the locality structure. 
Interesting future work includes finite block length considerations, and extending the scheme to spatially-coupled LDPC codes, as we pursued in \cite{RamCassuto18, RamCassuto19}.

\appendices 
\section{Proof of Theorem~\ref{th:2D DE}}
\label{App:2D DE}
We use mathematical induction on the number of iterations $l$. For $l=0$ the probability that a local or joint edge carries an erasure as a variable-to-check message is $\epsilon$. Since  $x_{-1}(\epsilon)=y_{-1}(\epsilon)=1$, then \eqref{eq:2D DE1}-\eqref{eq:2D DE2} with $l=0$ leads to
\begin{align*} 
\begin{split}
&x_0(\epsilon) = \epsilon \cdot
\lambda_L\left(1- \rho_L\left(1-1\right) \right)  \cdot
\Lambda_J\left(1- \rho_J\left(1-1\right) \right)= \epsilon,\\
&y_0(\epsilon) = \epsilon \cdot
\Lambda_L\left(1- \rho_L\left(1-1\right) \right)  \cdot
\lambda_J\left(1- \rho_J\left(1-1\right) \right)= \epsilon,
\end{split}
\end{align*}
hence, \eqref{eq:2D DE1}-\eqref{eq:2D DE2} hold for $l=0$. Assume correctness of \eqref{eq:2D DE1}-\eqref{eq:2D DE2} for some $l\geq0$, and consider iteration $l+1$. Recall that over the BEC, a check node will not send an erasure message on an outgoing edge if and only if all of its incoming edges carry valid bits 0/1. Let $u_{l}(\epsilon)$ and $w_{l}(\epsilon)$ designate the probability that an outgoing message from a local and joint check node, respectively, is an erasure. Then,
\begin{align} \label{eq:2D DE proof u_l}
\begin{split}
1-u_{l+1}(\epsilon) = \rho_L\left(1-x_l(\epsilon)\right),\\
1-w_{l+1}(\epsilon) = \rho_J\left(1-y_l(\epsilon)\right).
\end{split}
\end{align}
In addition, a variable node will send an erasure message on an outgoing edge if and only if all of its incoming edges are erased. This leads to 
\begin{align} \label{eq:2D DE proof x_l}
\begin{split}
x_{l+1}(\epsilon)  = \epsilon \cdot \lambda_L\left(u_{l+1}(\epsilon) \right)\cdot\Lambda_J\left(w_{l+1}(\epsilon) \right),\\
y_{l+1}(\epsilon)  = \epsilon \cdot \Lambda_L\left(u_{l+1}(\epsilon) \right)\cdot\lambda_J\left(w_{l+1}(\epsilon) \right).
\end{split}
\end{align} 
Combining \eqref{eq:2D DE proof u_l} and \eqref{eq:2D DE proof x_l} yields that \eqref{eq:2D DE1}-\eqref{eq:2D DE2} hold for $l+1$, thus by mathematical induction it holds for every $l\geq 0$. 

\section{Proof of Lemma~\ref{lemma:fix pt}} \label{App:fix pt}
\begin{enumerate}
\item Assume that $x=0$. Since $\Lambda_L(0)=0$, \eqref{eq:g} implies that $y=g(\epsilon,0,y)=0$. Moreover, if $y=0$ and $P_0=0$, then \eqref{eq:f} yields $x=f(\epsilon,x,0)=0$. Finally, if $y=0$ and $\lambda_J(0)>0$, then \eqref{eq:g} implies that $0=\Lambda_L(1-\rho_L(1-x))$, hence $\rho_L(1-x)=1$ and $x=0$. 
\item Follows immediately from \eqref{eq:f}, \eqref{eq:g} and \eqref{eq:f,g fixed}.
\item We prove \eqref{eq:no jumps} by a mathematical induction. For $l=0$, \eqref{eq:no jumps} holds due to Item~\ref{item: fix pt lemma1} and the fact that $x_{0} =y_{0}=\epsilon $. Assume correctness of \eqref{eq:no jumps} for some $l \geq 0$ and consider iteration $l+1$. In view of Lemma~\ref{lemma:monotonicity}, eq.~\eqref{eq:2D DE with f,g} and the induction assumption, it follows that
\begin{align} \label{eq:proof fix pt lemma}
\begin{split}
x_{l+1} = f \left(\epsilon,x_{l},y_{l}\right) \geq f(\epsilon,x,y)= x , \\
y_{l+1} = g\left(\epsilon,x_{l},y_{l}\right) \geq g(\epsilon,x,y)= y.
\end{split}
\end{align}
This prove correctness of \eqref{eq:no jumps} for $l+1$ and by mathematical induction proves \eqref{eq:no jumps} for all $l \geq 0$.
\end{enumerate}

\section{Proof of Lemma~\ref{lemma:q_L(0)}} \label{App:q_L(0)}
Let $I=\min \{i \colon \Lambda_{L,i}>0\}$. Since $\Lambda_L(0)=0$, then $I \geq 1$. The connection between $\lambda_L(\cdot)$ and $\Lambda_L(\cdot)$ in \eqref{eq:node-edge} yields

\begin{align}
\lim\limits_{u \to 0}\frac{\Lambda_L(u)}{\lambda_L(u)}
&=\lim\limits_{u \to 0}\frac{\Lambda_L(u)}{\Lambda_L'(u)}\cdot\Lambda'(1)\notag\\
&=\Lambda'(1)\lim\limits_{u \to 0}\frac{\sum_{i\geq I}\Lambda_{L,i} u^i}{\sum_{i\geq I}i\Lambda_{L,i} u^{i-1}} \notag\\ 
&=\Lambda'(1)\lim\limits_{u \to 0}\frac{u^I\sum_{i\geq I}\Lambda_{L,i} u^{i-I}}{u^{I-1}\sum_{i\geq I}i\Lambda_{L,i} u^{i-I}} \notag\\
&=\Lambda'(1)\frac{\Lambda_{L,I} }{I\Lambda_{L,I} }\lim\limits_{u \to 0}u \notag\\
&=0.
\end{align}
Further, let $u(x)=1-\rho_L(1-x)$ and note that $\lim_{x\to 0} u(x)=0$. Thus,
\begin{align*}
\lim_{x \to 0} x \cdot \frac{\Lambda_L\left( 1- \rho_L\left(1-x \right)\right)}{\lambda_L\left( 1- \rho_L\left(1-x \right) \right)} =\lim_{x \to 0}x \cdot \lim_{x \to 0}\frac{\Lambda_L(u(x))}{\lambda_L(u(x))}=\lim_{x \to 0}x \cdot \lim_{u \to 0}\frac{\Lambda_L(u)}{\lambda_L(u)} = 0.
\end{align*}

\section{Proof of Theorem~\ref{th:numerical th}} \label{App:numerical th}

\begin{lemma} \label{lemma: fixed curve}
	If $(x,y)$ is an $(f,g)$-fixed point with $y>0$, then $x \leq q(y)$.
\end{lemma}

\begin{proof}
	Let $\epsilon \in (0,1)$ and let $(x,y)$ be a solution to \eqref{eq:f,g fixed} with $y>0$. In view of \eqref{eq:f} and \eqref{eq:g}, dividing the first equation of \eqref{eq:f,g fixed} with the second one yields
	\begin{align}\label{eq:divide f,g fixed}
	\frac{x}{y} = \frac{\lambda_L\left(1- \rho_L\left(1-x \right)\right) \cdot \Lambda_J\left(1- \rho_J\left(1-y\right)\right)}{\lambda_J\left(1- \rho_J\left(1-y \right)\right) \cdot \Lambda_L\left(1- \rho_L\left(1-x\right)\right)}
	\end{align} 
	which after some rearrangements implies 
	\begin{align}\label{eq:q_L(x)=q_J(y)}
	q_J(y)=q_L(x),
	\end{align}\
	where $q_J(\cdot)$ and $q_L(\cdot)$ are defined in \eqref{eq:q_L q_J}.
	In view of \eqref{eq:q_L q_J}, since $(x,y)$ is an $(f,g)$-fixed point, then 
	\begin{align}\label{eq:q_J leq 1}
	q_J(y)&=y\cdot \frac{\Lambda_J\left( 1- \rho_J\left(1-y \right)\right)}{\lambda_J\left( 1- \rho_J\left(1-y \right)\right)} \notag \\
	&=g(\epsilon,x,y)\cdot \frac{\Lambda_J\left( 1- \rho_J\left(1-y \right)\right)}{\lambda_J\left( 1- \rho_J\left(1-y \right)\right)} \notag \\
	&=\overbrace{\epsilon}^{\leq 1}\cdot \overbrace{\Lambda_J\left( 1- \rho_J\left(1-y \right)\right)}^{\leq 1} \cdot \overbrace{\Lambda_L\left( 1- \rho_L\left(1-x \right)\right)}^{\leq 1}\notag \\
	&\leq 1,
	\end{align}
	which together with Definition~\ref{def:q} and \eqref{eq:q_L(x)=q_J(y)} completes the proof.
	
\end{proof}

Let
\begin{align} \label{eq:converse 01}
\epsilon >  \inf\limits_{\substack{y \in (0,1]  \\ q_J(y)\leq 1}} \frac{y}{g(1,q(y),y)}.
\end{align}
There exists $y_0 \in (0,1]$ such that $q_J(y_0) \leq 1$ and
\begin{align}\label{eq:y_0}
y_0=\epsilon\cdot g(1,q(y_0),y_0)= \epsilon \cdot \lambda_J(1-\rho_J(1-y_0))\cdot\Lambda_L(1-\rho_L(1-q(y_0))).
\end{align}
In view of \eqref{eq:q def}, 
\begin{align} 
q_L(q(y_0))=q_J(y_0),
\end{align}
which combined with \eqref{eq:f} and \eqref{eq:q_L q_J} yields
\begin{align} \label{eq:using q}
q(y_0) &= \frac{\lambda_L(1-\rho_L(1-q(y_0)))}{\Lambda_L(1-\rho_L(1-q(y_0)))} \cdot y_0 \cdot \frac{\Lambda_J(1-\rho_J(1-y_0))}{\lambda_J(1-\rho_J(1-y_0))} \notag \\
&=\epsilon \cdot \lambda_L(1-\rho_L(1-q(y_0)))\cdot \Lambda_J(1-\rho_J(1-y_0)) \notag \\
&= f(\epsilon,q(y_0),y_0).
\end{align}
Thus, $(q(y_0),y_0)$ is a non-zero $(f,g)$-fixed point, which in view of Theorem~\ref{theorem:fix pt char} implies that $\epsilon > \epsilon^*_G$. Hence,  
\begin{align} \label{eq:th 1st part}
\epsilon^*_G \leq \inf\limits_{\substack{y \in (0,1]  \\ q_J(y)\leq 1}} \frac{y}{g(1,q(y),y)}.
\end{align}

Next, let 
\begin{align} \label{eq:direct}
\epsilon <  \inf\limits_{\substack{y \in (0,1]  \\ q_J(y)\leq 1}} \frac{y}{g(1,q(y),y)}
\end{align}
and let $(x,y)$ be a solution to \eqref{eq:f,g fixed}. In what follows, we prove that $y=0$. Assume to the contrary that $y>0$. From Lemma~\ref{lemma: fixed curve} it follows that $x \leq q(y)$, which in view Lemma~\ref{lemma:monotonicity}, \eqref{eq:q_J leq 1} and \eqref{eq:direct} implies 
 \begin{align}\label{eq:numerical proof1}
y=g(\epsilon,x,y)\leq g(\epsilon,q(y),y)  < y,
\end{align}
in contradiction; thus, $y=0$. Next, consider two cases:
\begin{enumerate}
\item If $P_0=0$ or $\lambda_J(0) >0$, then Item~\ref{item: fix pt lemma0} of Lemma~\ref{lemma:fix pt}  implies that $x=0$. Hence, every $(f,g)$-fixed point satisfies $y=x=0$. In view of Theorem~\ref{theorem:fix pt char}, it follows that if \eqref{eq:direct} holds, then $\epsilon < \epsilon^*_G$, so
\begin{align}
\epsilon^*_G \geq  \inf\limits_{\substack{y \in (0,1]  \\ q_J(y)\leq 1}} \frac{y}{g(1,q(y),y)}
\end{align} 
which with \eqref{eq:th 1st part} completes the proof when $P_0=0$ or $\lambda_J(0) >0$. 

\item If $P_0>0$ and $\lambda_J(0) =0$, it is not true in general that for every fixed point $(x,y)$, $y=0$ implies $x=0$. However, if in addition to \eqref{eq:direct},
\begin{align} \label{eq:direct2}
\epsilon <  \frac1{P_0}\cdot \inf_{(0,1]}\frac{x}{\lambda_L\left(1- \rho_L\left(1-x \right)\right)},
\end{align}
and $y=0$ for some fixed point $(x,y)$, then $x=0$. To see this, assume to the contrary that $x>0$. In view of \eqref{eq:f} and \eqref{eq:direct2} it follows that
\begin{align} \label{eq:x fixed}
x= f(\epsilon,x,0) = \epsilon \cdot P_0 \cdot \lambda_L\left(1- \rho_L\left(1-x \right)\right) < x
\end{align}
in contradiction; hence, if \eqref{eq:direct} and \eqref{eq:direct2} hold, $x=0$ thus $\epsilon < \epsilon^*_G$. This means that 
\begin{align}\label{eq:direct proof}
\epsilon^*_G \geq \min\left\{ \inf\limits_{\substack{y \in (0,1]  \\ q_J(y)\leq 1}} \tfrac{y}{g(1,q(y),y)},\;\;\tfrac1{P_0}\cdot \inf_{(0,1]}\tfrac{x}{\lambda_L\left(1-  \rho_L\left(1-x \right)\right)}\right\}.
\end{align}
To complete the proof, we must show that when $P_0>0$ and $\lambda_J(0) =0$, then
\begin{align}\label{eq:converse proof}
\epsilon^*_G \leq \min\left\{ \inf\limits_{\substack{y \in (0,1]  \\ q_J(y)\leq 1}} \tfrac{y}{g(1,q(y),y)},\;\;\tfrac1{P_0}\cdot \inf_{(0,1]}\tfrac{x}{\lambda_L\left(1-  \rho_L\left(1-x \right)\right)}\right\}.
\end{align}
If
\begin{align*}
\inf\limits_{\substack{y \in (0,1]  \\ q_J(y)\leq 1}} \frac{y}{g(1,q(y),y)} \leq \frac1{P_0}\inf_{(0,1]}\frac{y}{\lambda_L\left(1-  \rho_L\left(1-y \right)\right)},
\end{align*}
then \eqref{eq:converse proof} follows immediately from \eqref{eq:th 1st part}; hence we can assume that 
\begin{align} \label{eq:converse1}
\inf\limits_{\substack{y \in (0,1]  \\ q_J(y)\leq 1}} \frac{y}{g(1,q(y),y)} > \frac1{P_0}\inf_{(0,1]}\frac{y}{\lambda_L\left(1-  \rho_L\left(1-y \right)\right)}.
\end{align}
Let 
\begin{align} \label{eq:converse12}
\inf\limits_{\substack{y \in (0,1]  \\ q_J(y)\leq 1}} \frac{y}{g(1,q(y),y)} >\epsilon> \frac1{P_0}\inf_{(0,1]}\frac{y}{\lambda_L\left(1-  \rho_L\left(1-y \right)\right)},
\end{align}
and let $x_0 \in (0,1]$, such that $x_0=\epsilon \cdot P_0 \cdot \lambda_L\left(1-  \rho_L\left(1-x_0 \right)\right)$. Since $\lambda_J(0)=0$, it follows that $(x_0,0)$ is a fixed point with $x_0>0$, thus $\epsilon > \epsilon^*_G$. Since this is true for every $\epsilon> \tfrac1{P_0}\inf_{(0,1]}\frac{y}{\lambda_L\left(1-  \rho_L\left(1-y \right)\right)}$, then $\epsilon^*_G \leq \tfrac1{P_0}\inf_{(0,1]}\frac{y}{\lambda_L\left(1-  \rho_L\left(1-y \right)\right)}$. In view of \eqref{eq:converse1}, it follows that \eqref{eq:converse proof} holds. This completes the proof for the $P_0>0$ and $\lambda_J(0) =0$ case.
\end{enumerate}

\section{Proof of Lemma~\ref{lemma:valid opt schedule}} \label{App:valid opt schedule}
To prove Lemma~\ref{lemma:valid opt schedule} we need the following lemma.

\begin{lemma}\label{lemma:valid schedule}
A scheduling scheme is valid if and only if, $\epsilon_{loc}(y_l)<\epsilon_L$, for some iteration $l$.
\end{lemma}

\begin{proof}
Recall the definition of the local threshold,
\begin{align} \label{eq:loc th def}
\epsilon_L = \sup \{ \epsilon \colon x=\epsilon\lambda_L(1-\rho_L(1-x)) \text{ has no solution in } (0,1]\},
\end{align}
and let $x=\lim\limits_{l \to \infty}x_l$ and $y=\lim\limits_{l \to \infty}y_l$.
Since under every scheduling scheme $y_l$ is monotonically non-increasing in $l$, then in view of \eqref{eq:f with eff},
\begin{align*}
\exists l \in \mathbb{N}, \quad\epsilon_{loc}(y_l)<\epsilon_L 
&\Leftrightarrow \epsilon_{loc}(y) < \epsilon_L \\
&\Leftrightarrow x=\epsilon_{loc}(y) \lambda_L(1-\rho_L(1-x)) \text{ has no solution for }x \in (0,1] \\
&\Leftrightarrow \lim\limits_{l \to \infty}x_l=0.
\end{align*} 
\end{proof}
We proceed with the proof of Lemma~\ref{lemma:valid opt schedule}.
Let $(\Lambda_L,\Lambda_J,\Omega_L,\Omega_J)$ be degree-distribution polynomials, and let $\epsilon\leq \epsilon^*_G(\Lambda_L,\Lambda_J,\Omega_L,\Omega_J)$. Let $(x,y)=\lim\limits_{l \to \infty}(x_l,y_l)$. Assume in contradiction that $\epsilon_{loc}(y) \geq \epsilon_L$. Since $\eta>0$, letting $l \to \infty$ in \eqref{eq:opt schedule} implies that $(x,y)$ is a non-trivial $(f,g)$-fixed point. However, in view of Theorem~\ref{theorem:fix pt char}, if $\epsilon\leq \epsilon^*_G(\Lambda_L,\Lambda_J,\Omega_L,\Omega_J)$, every $(f,g)$-fixed point is the trivial point, in contradiction. Thus,  $\epsilon_{loc}(y) < \epsilon_L$ which, due to Lemma~\ref{lemma:valid schedule}, completes the proof.

\section{Proof of Lemma~\ref{lemma:opt schedule}} \label{App:opt schedule}
Let $\left\{l^{(1)}_k\right\}_{k=1}^{N^{(1)}_{JI}}$ and  $\left\{l^{(2)}_k\right\}_{k=1}^{N^{(2)}_{JI}}$ be the joint-update iterations of the scheduling scheme described in \eqref{eq:x_l_k} and in some arbitrary valid scheduling scheme, receptively. We need to show that $N_{JI}^{(1)} \leq N_{JI}^{(2)}$. 
To proceed we need the following lemmas:
\begin{lemma}\label{lemma:x_s monotonic}
	$ x_s(\epsilon) $ as defined in Definition~\ref{def:h_eps} is monotonic non-decreasing in $ \epsilon\;. $
\end{lemma}
\begin{proof}
	Let $ \epsilon_1\leq \epsilon_2 $, and consider $ x_s(\epsilon_1),\;x_s(\epsilon_2)\;. $ In view of Definition~\ref{def:h_eps}, 
	\begin{align*}
		h_{\epsilon_2}\left (x_s(\epsilon_1)\right ) 
		&\triangleq \epsilon_2\lambda_L(1-\rho_L(1-x_s(\epsilon_1)))-x_s(\epsilon_1)\\
		&\geq		\epsilon_1\lambda_L(1-\rho_L(1-x_s(\epsilon_1)))-x_s(\epsilon_1)\\
		&\triangleq	h_{\epsilon_1}\left (x_s(\epsilon_1)\right )\\
		&\geq 		0\;. 
	\end{align*}
	Thus, $ x_s(\epsilon_2)\triangleq \max\{x\in[0,1]\;\colon h_{\epsilon_2}\left (x\right )\geq 0\}\geq x_s(\epsilon_1)\;. $
\end{proof}
\begin{lemma} \label{lemma:opt schdule induction}
Let
\begin{align}\label{eq:eps_k 1 and 2}
\begin{split}
&\varepsilon^{(1)}_k=\epsilon_{loc}\left(y_{l^{(1)}_k}\right),\quad 1\leq k \leq N^{(1)}_{JI}, \\
&\varepsilon^{(2)}_k=\epsilon_{loc}\left(y_{l^{(2)}_k}\right),\quad 1\leq k \leq N^{(2)}_{JI}.
\end{split}
\end{align}
Then, for every $ 1 \leq k \leq \min\left(N^{(1)}_{JI},N^{(2)}_{JI} \right) $,
\begin{align} \label{eq:opt schdule induction}
y_{l^{(1)}_k} \leq y_{l^{(2)}_k},\quad\text{and}\quad \varepsilon^{(1)}_k\leq \varepsilon^{(2)}_k,\quad\text{and}\quad x_{l^{(1)}_k} \leq x_{l^{(2)}_k}\;.
\end{align}
\end{lemma}

\begin{proof}
By induction on $1 \leq k \leq \min\left(N^{(1)}_{JI},N^{(2)}_{JI} \right)$. 
In the first joint update, we have $y_{l^{(1)}_1}=1=y_{l^{(2)}_1}$ and $ \varepsilon_1\triangleq \epsilon_{loc}\left (y_{l^{(1)}_1}\right )=\epsilon$. Thus, in view of Definition~\ref{def:h_eps}, in the first joint update $ =x_{l^{(2)}_1}\geq x_s(\epsilon)=x_{l^{(1)}_1}$.
Hence \eqref{eq:opt schdule induction} holds for $k=1$. Assume correctness for some joint update $k<\min\left(N^{(1)}_{JI},N^{(2)}_{JI} \right)$, and consider update $k+1$. In view of Lemma~\ref{lemma:monotonicity}, \eqref{eq:y_l_k}-\eqref{eq:x_l_k}, and the induction assumption,
$y_{l^{(1)}_{k+1}}= g\left(\epsilon,x_{l^{(1)}_{k}},y_{l^{(1)}_{k}}\right) \leq g\left(\epsilon,x_{l^{(2)}_{k}},y_{l^{(2)}_{k}}\right) =y_{l^{(2)}_{k+1}}\;,$
which together with \eqref{eq:eps_k} implies that $ \varepsilon^{(1)}_{k+1}\triangleq \epsilon_{loc}\left (y_{l^{(1)}_{k+1}}\right )\leq \epsilon_{loc}\left (y_{l^{(2)}_{k+1}}\right )\triangleq\varepsilon^{(2)}_{k+1}.$ In view of Lemma~\ref{lemma:x_s monotonic}, it follows that $ x_{l^{(1)}_{k+1}}\triangleq x_s(\varepsilon^{(1)}_{k+1})\leq x_s(\varepsilon^{(2)}_{k+1})\leq x_{l^{(2)}_{k+1}}.$ By induction, we complete the proof.
\end{proof}
We proceed with the proof of Lemma~\ref{lemma:opt schedule}.
Assume, on the contrary, that  $N_{JI}^{(1)} > N_{JI}^{(2)}$. Lemma~\ref{lemma:opt schdule induction} and the monotonicity of $ \varepsilon_k $ in $ k $ imply that
\begin{align} \label{eq:eps^1 > eps_L}
\epsilon^{(2)}_{N_{JI}^{(2)}} \geq \epsilon^{(1)}_{N_{JI}^{(2)}}  \geq \epsilon^{(1)}_{N_{JI}^{(1)}-1} \geq \epsilon_L,
\end{align}
which, in view of Lemma~\ref{lemma:valid schedule} yields that the scheduling scheme indexed by $\left\{l^{(2)}_k\right\}_{k=1}^{N^{(2)}_{JI}}$ is not valid, in contradiction. Thus, $N_{JI}^{(1)} \leq N_{JI}^{(2)}$.

\section{Proof of Lemma~\ref{lemma:x_s c.a.}} \label{App:x_s c.a.}
In view of Definition~\ref{def:h_eps}, let $h^{(k)}_\epsilon(x)=\epsilon\lambda^{(k)}(1-\rho^{(k)}(1-x))-x$. Eq. \eqref{eq:c.a. known} yields
\begin{align}\label{eq:c.a. h_eps}
h_\epsilon(x)=\lim_{k \to \infty}h^{(k)}_\epsilon(x)= \left\{ 
\begin{array}{ll}
\left( \frac{\epsilon}{\epsilon_L}-1\right) x & 0 \leq x \leq \epsilon_L \\
\epsilon- x & \epsilon_L \leq x \leq \epsilon 
\end{array}
\right.
\end{align}
For every $k \in \mathbb{N}$, and $x\in (\epsilon,1]$, 
\begin{align*}
h_\epsilon^{(k)}(x) 
&= \epsilon\lambda^{(k)}(1-\rho^{(k)}(1-x))-x \\ \notag
&\leq \epsilon-x \\ \notag
&<0,
\end{align*}
Thus
\begin{align}\label{eq:x_s^k <eps}
x_s^{(k)}(\epsilon) \leq \epsilon ,\quad \forall k \in \mathbb{N}.
\end{align}
In addition, for every $0<a<\epsilon$ there exists $K_0$ such that 
\begin{align*}
h_\epsilon^{(k)}(\epsilon-a)>0 ,\quad \forall k \geq K_0,
\end{align*}
so $x_s^{(k)}(\epsilon) \geq \epsilon-a$, for every $ k \geq K_0$; hence,
\begin{align} \label{eq:x_s liminf}
\liminf_{k \to \infty} x_s^{(k)}(\epsilon) \geq \epsilon.
\end{align}
Combining \eqref{eq:x_s^k <eps} and \eqref{eq:x_s liminf} implies that $\lim\limits_{k \to \infty}x^{(k)}_s(\epsilon)$ exists, and completes the proof.

\end{document}